\documentclass[10pt,reqno]{amsart}

\usepackage{amsmath}
\usepackage{amssymb}
\usepackage{amsfonts}
\usepackage{graphicx}
\usepackage{color}
\usepackage{times}
\usepackage{mathptmx} 
\usepackage{datetime}
\usepackage{framed}

\usepackage{datetime}
\usepackage{times} 
\usepackage{framed} 
\usepackage{graphicx} 
\usepackage{amsmath} 
\usepackage{amssymb}  
\newtheorem{lem}{Lemma}
\newtheorem{thm}{Theorem}

\def\fa{\mathfrak{a}}

\def\<{\leqslant}           
\def\>{\geqslant}           
\def\div{\mathrm{div}}         
\def\ad{\mathrm{ad}}           

\def\d{\partial}

\def\wt{\widetilde}

\def\Re{\mathrm{Re\, }}   
\def\Im{\mathrm{Im\, }}   

\def\cH{\mathcal{H}}   
\def\mA{\mathbb{A}}    
\def\mZ{\mathbb{Z}}    
\def\mR{\mathbb{R}}    
\def\mC{\mathbb{C}}    

\def\Tr{\mathrm{Tr}}       
\def\rT{\mathrm{T}}        
\def\rF{\mathrm{F}}        
\def\sinc{\mathrm{sinc}}        

\def\bE{\mathbf{E}}    

\def\[[[{[\![\![}
\def\]]]{]\!]\!]}

\def\bra{{\langle}}
\def\ket{{\rangle}}

\def\Bra{\left\langle}
\def\Ket{\right\rangle}

\def\dbra{\langle\!\!\langle}
\def\dket{\rangle\!\!\rangle}

\def\re{\mathrm{e}}        
\def\rd{\mathrm{d}}        

\def\cL{\mathcal{L}}

\def\bR{\mathbf{R}}

\def\bJ{\mathbf{J}}

\def\x{\times}
\def\ox{\otimes}

\def\fA{\mathfrak{A}}
\def\fB{\mathfrak{B}}

\def\fF{\mathfrak{F}}
\def\fG{\mathfrak{G}}
\def\fH{\mathfrak{H}}

\def\cF{\mathcal{F}}
\def\cW{\mathcal{W}}

\def\cD{\mathcal{D}}

\def\cG{\mathcal{G}}
\def\cI{\mathcal{I}}

\def\cA{\mathcal{A}}
\def\cB{\mathcal{B}}
\def\cE{\mathcal{E}}
\def\cov{\mathbf{cov}}

\def\cS{\mathcal{S}}

\def\mH{\mathbb{H}}
\def\mS{\mathbb{S}}

\def\mZ{\mathbb{Z}}

\def\eps{\epsilon}
\def\Ups{\Upsilon}
\def\ups{\upsilon}

\def\rprod{\mathop{\overrightarrow{\prod}}}

\def\sn{|\!|\!|}

\begin{document}
\title[Quasi-characteristic functions
in quantum stochastic systems
]{Evolution of 
quasi-characteristic functions
in 
quantum stochastic systems with Weyl quantization of 
energy operators}


\author{Igor G. Vladimirov}
\address{UNSW Canberra, ACT 2600, Australia}
\email{igor.g.vladimirov@gmail.com}
\thanks{This work is supported by the Australian Research Council}

\subjclass[2010]{Primary:
81S22, 
81S25, 
81S30, 
81P16, 
81S05; 
secondary:
81Q15, 
35Q40, 
37M25. 
}

\keywords{Open quantum systems, canonical commutation relations, quantum stochastic differential equations, Weyl quantization, Wigner-Moyal approach, quasi-characteric functions, quasi-probability density functions, integro-differential equations.
}


\begin{abstract}
This paper considers open quantum systems whose dynamic variables  satisfy  ca\-nonical commutation relations and are governed by Markovian  Hudson-Par\-tha\-sarathy  quantum stochastic differential equations driven by external bosonic fields. The dependence of the Hamiltonian and the  system-field coupling operators on the system variables is represented using the Weyl functional calculus. This leads to an integro-differential equation (IDE) for the evolution of the quasi-characteristic function (QCF) which  encodes the dynamics of mixed moments of the system variables. Unlike quantum master equations  for reduced density operators,   this IDE involves only complex-valued functions on finite-dimensional Euclidean spaces and extends the Wigner-Moyal phase-space approach for quantum stochastic systems. The dynamics of the QCF and the related Wigner quasi-probability density function (QPDF) are discussed in more detail for the case when the coupling operators depend linearly on the system variables and the Hamiltonian has a nonquadratic part represented in the Weyl quantization form.
For this class of quantum stochastic systems, we also consider an approximate computation of invariant states and discuss the deviation from Gaussian quantum states in terms of the $\chi^2$-divergence (or the second-order Renyi relative entropy) applied to the QPDF.
The results of the paper may find applications to investigating different aspects of the moment stability, relaxation dynamics and invariant states in open quantum systems.
\end{abstract}

\maketitle

\thispagestyle{empty}




\section{Introduction}\label{sec:intro}

A wide class of open quantum systems, whose dynamics are affected by interaction with the environment and are described in terms of noncommutative operators on a Hilbert space evolving according to the laws of quantum mechanics, can be modelled by using the Hudson-Parthasarathy quantum stochastic calculus \cite{HP_1984,P_1992}; see also \cite{B_2010,H_1991}. This approach represents the external bosonic fields by annihilation and creation processes (which constitute a quantum mechanical counterpart to the classical Wiener process \cite{KS_1991}) and gauge processes associated with photon exchange between the fields. The continuous tensor product structure of the symmetric Fock space \cite{PS_1972}, which serves as a domain for the field operators,  and the role of the quantum Wiener process as an innovation process are important ingredients of a Markovian model of the system dynamics.
This model follows the Heisenberg picture of quantum dynamics
\cite{M_1998} in the form of quantum stochastic differential equations (QSDEs) for the system variables, which are driven by the field operator processes according to the energetics of the system-field interaction. This interaction is specified by the Hamiltonian, coupling and scattering operators which are (in general, nonlinear) functions of the system variables.

The fact, that the structure of QSDEs reflects the joint unitary evolution of the system and fields and  is dictated by the energy operators, underlies the interconnection rules for open quantum systems in quantum feedback networks \cite{GJ_2009} and is responsible for physical realizability constraints \cite{JNP_2008,SP_2012} in coherent quantum control and filtering problems \cite{MJ_2012,NJP_2009,VP_2013a,VP_2013b}. These problems are measurement-free versions of  the measurement-based control and filtering problems for quantum systems \cite{B_1983,DDJW_2006,EB_2005,GBS_2005,J_2005,YB_2009} and aim to achieve desired properties for (or extracting quantum information from) a given quantum system through its interconnection with another quantum system, which plays the role of a controller or observer and replaces the classical observation-actuation loop. Similarly to their measurement-based counterparts, the coherent quantum control and filtering problems employ performance criteria associated with the averaged behaviour of the resulting fully quantum systems. Such performance functionals are organised as quantum expectations of nonlinear (for example, positive definite quadratic or quadratic-exponential) functions of system variables which are subject to  minimization, thus reflecting a preference towards dissipativity of the quantum system with respect to external disturbances \cite{JG_2010,PUJ_2012,VP_2012b} in the spirit of the  Lyapunov stability and Willems dissipativity theories \cite{W_1972}.

Therefore, the above mentioned control and filtering problems employ generalized moments  which may involve nonlinear (but not necessarily polynomial) functions of the system variables.
These moments are completely specified by the mean vector and the quantum covariance matrix of the system variables in the case of Gaussian quantum states \cite{P_2010},  the class of which is invariant with respect to linear quantum dynamics \cite{JK_1998} of open quantum harmonic oscillators \cite{EB_2005,GZ_2004}. An example of tractable non-Gaussian moment dynamics is provided by quasi-linear quantum stochastic systems \cite{VP_2012c}. The generation of specific classes of Gaussian and non-Gaussian states in appropriately engineered  quantum systems (in particular, using quantum-optical components) and criteria for the existence and stability of invariant states are a subject of research \cite{MWPY_2014,PAMGUJ_2014,Y_2009,ZJ_2011}.

In general, the moments of the system variables are encoded in their quasi-characteristic function (QCF) \cite{CH_1971}, and it is the QCF evolution that is the main theme of the present paper. More precisely, we are concerned with open quantum systems whose dynamic variables satisfy canonical commutation relations (CCRs), similar to those of the position and momentum operators, and are governed by Markovian QSDEs with the identity scattering matrix. Furthermore, the dependence of the Hamiltonian and the system-field coupling operators on the system variables is represented by  using the Weyl functional calculus \cite{F_1989}. The Weyl quantization form of the energy operators allows a linear integro-differential equation (IDE) to be obtained for the evolution of the QCF, which  encodes the moment dynamics of the system. The resulting   IDE  is a quantum analogue of the corresponding equation for the characteristic functions of Markov diffusion processes, obtained in the classical case through the Fourier transform of the Fokker-Planck-Kolmogorov equation (FPKE) \cite{KS_1991,S_2008} into the spatial frequency domain.

Although this approach
to open quantum dynamics pertains to the Wigner-Moyal phase-space method \cite{H_2010,M_1949} of quasi-probability density functions  (QPDFs) \cite{GRS_2014,KS_2008} which are Fourier transforms of the QCFs, the contribution of the present study is in systematically combining, for this purpose, the structure of the Hudson-Parthasarathy QSDEs with the Weyl quantization model of the Hamiltonian and the coupling operators. The IDE, which  governs the evolution of the QCF, can be regarded as a spatial frequency domain representation of the master equations for reduced density operators  \cite{GZ_2004}, known for particular classes of quantum systems such as the open quantum harmonic oscillators mentioned above. Unlike the quantum master equations and similarly to the Moyal equations \cite{M_1949}, the IDE for the QCF involves only complex-valued functions on finite-dimensional Euclidean spaces and its analysis can be more convenient from the viewpoint of classical partial differential equations (PDEs) \cite{E_1998,V_1971}. We also mention recent extensions of the Moyal equations to different classes of open quantum systems in \cite{GRS_2015,MD_2015}.

As an illustration of the phase-space analysis in the quantum stochastic framework, we discuss a class \cite{SVP_2014} of open quantum systems with linear coupling to the external bosonic fields, in which case, the nonlinearity in the governing QSDE is caused by a nonquadratic part of the system Hamiltonian represented in the Weyl quantization form. For such a system, the QPDF satisfies an IDE consisting of an FPKE part (which corresponds to a linear SDE leading to Gaussian dynamics) and an integral operator term. This integral operator does not correspond to the ``jump'' part of a classical jump-diffusion process and can lead to negative values of the QPDF, 
which makes it qualitatively different  from usual PDFs and is considered to be a resource provided by quantum systems in comparison with their classical  counterparts \cite{VFGE_2012}. 
In the case of linear system-field coupling, we also discuss a dissipation relation for a weighted $L^2$-norm of the QCF, which is organised similarly to the norm in the Bessel potential space \cite{S_2008} and can be applied to obtaining upper bounds for the QPDF and its derivatives. Furthermore, we consider a perturbative  computation of the invariant state in phase space as a steady-state solution  of the IDEs for the QCF and QPDF  through the operator splitting \cite{Mar_1988,S_1968}. Also, for the case of linear system-field coupling, we discuss a dissipation relation  for the deviation of the system from Gaussian quantum states in terms of the $\chi^2$-divergence (or the second-order Renyi relative entropy \cite{R_1961}) of the QPDFs.

In addition to these examples, the results of the paper can be used for investigating the moment stability and the rate of convergence to invariant states in open quantum stochastic systems, as well as other aspects of the relaxation dynamics.
We omit some analytic details (such as regularity issues),  so that the present exposition  is fairly intuitive and maintains a ``physical'',  rather than ``mathematical'', level of rigour.

The paper is organised as follows. Section~\ref{sec:not} provides principal notation for convenience of reading. Section~\ref{sec:system} describes the class of open quantum systems being considered.
Section~\ref{sec:weyl} specifies the Weyl quantization of the Hamiltonian and the coupling operators and represents the governing QSDE in a similar form.
Section~\ref{sec:classlim} discusses the classical limit of this equation which corresponds to a commutative Markov diffusion process governed by a Hamiltonian  SDE with a canonical flow in the sense of \cite{G_1999,G_2015}.
Section~\ref{sec:quasi} revisits the generalized moments of system variables in terms of the QCFs and QPDFs.
Section~\ref{sec:evol} obtains the IDE which governs the evolution of QCF.
Section~\ref{sec:Hlin} discusses this equation together with a related IDE for the QPDF and a dissipation relation for the QCF for the class of systems with  linear system-field coupling.
Section~\ref{sec:inv} considers an approximate phase-space computation of the invariant state for such systems with a nonquadratic potential.
 Section~\ref{sec:chi2} applies the above results to the deviation of the system from Gaussian states in terms of the $\chi^2$-divergence of QPDFs. 
Section~\ref{sec:conc} provides concluding remarks.

\section{Notation}\label{sec:not}

The commutator of linear operators $A$ and $B$ is denoted by $[A,B]:= AB-BA$, with $\ad_A(\cdot):= [A,\cdot]$ being a linear superoperator associated with a fixed operator $A$.
This extends to the commutator $(n\x m)$-matrix  $
    [X,Y^{\rT}]\!
    :=\!
    XY^{\rT} - (YX^{\rT})^{\rT} \!=\! ([X_j,Y_k])_{1\< j\< n,1\< k\< m}
$ for a vector $X$ of operators $X_1, \ldots, X_n$ and a vector $Y$ of operators $Y_1, \ldots, Y_m$.  Vectors are organised as columns unless indicated otherwise,  and the transpose $(\cdot)^{\rT}$ acts on matrices of operators as if their entries were scalars.
In application to such matrices, $(\cdot)^{\dagger}:= ((\cdot)^{\#})^{\rT}$ denotes the transpose of the entrywise operator adjoint $(\cdot)^{\#}$, with  $(\cdot)^{\dagger}$ reducing to the usual complex conjugate transpose  $(\cdot)^*:= (\overline{(\cdot)})^{\rT}$ for complex matrices.
The subspaces of real symmetric, real antisymmetric and complex Hermitian matrices of order $n$ are denoted by $\mS_n$, $\mA_n$
 and
$
    \mH_n
    :=
    \mS_n + i \mA_n
$, respectively,  where $i:= \sqrt{-1}$ is the imaginary unit. The symmetrizer of a square matrix $M$ is defined by $\cS(M):= \frac{1}{2}(M+M^{\rT})$.
The real and imaginary parts of a complex matrix extend to matrices $M$ with operator-valued entries as $\Re M = \frac{1}{2}(M+M^{\#})$ and $\Im M = \frac{1}{2i}(M-M^{\#})$ which consist of self-adjoint operators. Positive (semi-) definiteness of matrices and the corresponding partial ordering  are denoted by ($\succcurlyeq$) $\succ$.  Also, $\mS_n^+$ and $\mH_n^+$ denote the sets of positive semi-definite real symmetric and complex Hermitian matrices of order $n$, respectively.
The tensor product of spaces or operators (in particular, the Kronecker product of matrices) is denoted by $\ox$. The identity matrix of order $n$ is denoted by $I_n$, while the identity operator on a linear space $H$ is denoted by $\cI_H$. 
Also, $\|v\|_K:= \sqrt{v^{\rT}Kv}$ denotes the (semi-) norm of a real vector $v$ associated with a real positive (semi-) definite symmetric matrix $K$. The Frobenius inner product of real or complex matrices is denoted by
$
    \bra M,N\ket_{\rF}
    :=
    \Tr(M^*N)
$ and generates the Frobenius norm $\|M\|_{\rF}:= \sqrt{\bra M,M\ket_{\rF}}$ which reduces to the standard Euclidean norm $|\cdot|$ for vectors.
At the same time, $\|\varphi\|_2:= \sqrt{\bra \varphi, \varphi\ket}$ denotes the norm in the Hilbert space $L^2(\mR^n)$ of square integrable complex-valued functions on $\mR^n$ with the inner product $\bra \varphi, \psi\ket:= \int_{\mR^n} \overline{\varphi(x)}\psi(x)\rd x$. The expectation $\bE \xi := \Tr(\rho \xi)$  of a quantum variable $\xi$ over a density operator $\rho$ extends entrywise to matrices of such variables. For vectors $X$ and $Y$ of quantum variables, $\cov(X,Y):= \bE (X Y^{\rT}) - \bE X \bE Y^{\rT}$ and $\cov(X):= \cov(X,X)$ denote the corresponding quantum covariance matrices. The ``rightwards'' ordered product of noncommutative variables $\xi_1, \ldots, \xi_n$ is denoted by $ \rprod_{k=1}^n \xi_k := \xi_1\x \ldots  \x \xi_n $. For a vector $X$ with entries $X_1,\ldots, X_n$ and an $n$-index $\alpha:=(\alpha_k)_{1\< k\< n}\in \mZ_+^n$ (where $\mZ_+$ denotes the set of nonnegative integers), use is made of the multiindex notation $X^{\alpha}:= \rprod_{k=1}^n X_k^{\alpha_k}$, $|\alpha|:= \alpha_1 + \ldots +\alpha_n$, $\alpha!:= \alpha_1! \x \ldots \x \alpha_n!$,  and $\d_u^{\alpha}:= \d_{u_1}^{\alpha_1}\ldots\d_{u_n}^{\alpha_n}$, where $\d_{u_1}, \ldots, \d_{u_n}$ are the partial derivatives with respect to independent real variables $u_1, \ldots, u_n$ comprising a vector $u:= (u_k)_{1\< k\< n}\in \mR^n$.
Use is also made the function $
    \sinc(z) := 
    \left\{{\small\begin{matrix}1 & {\rm if}\ z= 0\\
    \frac{\sin z}{z} & {\rm otherwise}\end{matrix}}\right.
$ (which is an entire even function of a complex variable). The divergence operator $\div(\cdot)$, when it is applied to an $\mR^{m\x n}$-valued function $M := (M_{jk})_{1\< j\< m, 1\< k\< n}$  on $\mR^n$ (with $m>1$), acts  in a row-wise fashion, with $\div M := \big(\sum_{k=1}^n \nabla_k M_{jk}\big)_{1\< j\< m}$ being an $\mR^m$-valued function, where $\nabla_k$ denotes the partial derivative with respect to the $k$th Cartesian coordinate.

\section{Open quantum stochastic systems}\label{sec:system}

We will be concerned with an open quantum stochastic system endowed with a vector $X:= (X_k)_{1\< k\< n}$ of dynamic variables $X_1,\ldots, X_n$. The  system variables are self-adjoint operators on an underlying complex separable Hilbert space $\cH$ which satisfy the Weyl CCRs
\begin{equation}
\label{CCR}
    \cW_{u+v} = \re^{i u^{\rT}\Theta v} \cW_u \cW_v
    =
     \re^{-i u^{\rT}\Theta v} \cW_v \cW_u
\end{equation}
for all $u,v\in \mR^n$, and hence,
$    [\cW_u, \cW_v]
    =
    -2i\sin(u^{\rT}\Theta v)\cW_{u+v}
$. Here, $\Theta:= (\theta_{jk})_{1\< j,k\< n}\in \mA_n$,
and use is made of the unitary Weyl operator
\begin{equation}
\label{cW}
  \cW_u := \re^{iu^{\rT} X}
\end{equation}
defined in terms of the self-adjoint operator $u^{\rT} X = \sum_{k=1}^n u_k X_k$ which is a linear combination of the system variables with real coefficients comprising the vector $u:=(u_k)_{1\< k\< n}$.  The Heisenberg infinitesimal form of the CCRs (\ref{CCR}) is
\begin{equation}
\label{Theta}
    [X, X^{\rT}]
     =
     2i \Theta \ox \cI_{\cH}
\end{equation}
on a dense domain in $\cH$. In what follows, the matrix $\Theta\ox \cI_{\cH}$ will be identified with $\Theta$.
Also, the dimension $n$ is assumed to be even, and the CCR matrix is given by
\begin{equation}
\label{ThetaJ}
    \Theta := \frac{1}{2}\bJ\ox I_{n/2},
    \qquad
    \bJ:=
    \begin{bmatrix}
        0 & 1\\
        -1 & 0
    \end{bmatrix}.
\end{equation}
This corresponds to the case when the vector $X$ is formed from $\frac{n}{2}$ conjugate pairs of the quantum mechanical position and momentum operators (with the units chosen so that the reduced Planck constant is $\hslash = 1$). However, the explicit form (\ref{ThetaJ}) of the CCR matrix  $\Theta$ will not be important, though the nonsingularity  $\det \Theta \ne 0$ will sometimes be used.
The vector $X$ of system variables evolves in time $t\> 0$ according to a particular yet important class of Markovian Hudson-Parthasarathy QSDEs \cite{HP_1984,P_1992} with the identity scattering matrix (which effectively eliminates from consideration the gauge processes associated with the photon exchange between the fields \cite{P_1992}):
\begin{equation}
\label{dX}
    \rd X =
    f\rd t + g\rd W,
\end{equation}
where the time arguments are omitted for brevity. The $n$-dimensional
drift vector $f$ and the dispersion $(n\x m)$-matrix $g$ of the QSDE (\ref{dX}) are given by
\begin{equation}
\label{fg}
    f := \cL(X) = i[h_0, X] + \cD(X),
    \qquad
    g := - i[X,h^{\rT}],
    \qquad
    h:=
    \begin{bmatrix}
        h_1\\
        \vdots\\
        h_m
    \end{bmatrix}.
\end{equation}
Here, $h_0$ is the system Hamiltonian and $h_1, \ldots, h_m$ are the system-field coupling operators. These are self-adjoint operators on the space $\cH$ which specify the energetics of the system and its interaction with the environment.
Furthermore, $\cL$ is the Gorini-Kossakowski-Sudar\-shan-Lindblad (GKSL) generator \cite{GKS_1976,L_1976}, which acts on a system operator $\xi$ as
\begin{equation}
\label{cL}
   \cL(\xi):= i[h_0,\xi] +\cD(\xi)
\end{equation}
and is evaluated entrywise at the vector $X$ in (\ref{fg}),
and $\cD$ is the decoherence superoperator given by
\begin{equation}
\label{cD}
    \cD(\xi)
    :=
    \frac{1}{2}
    \sum_{j,k=1}^m
    \omega_{jk}
    \big( [h_j,\xi]h_k + h_j[\xi,h_k]\big)
    =
    \frac{1}{2}
    \big(
        [h^{\rT},\xi]\Omega h  + h^{\rT}\Omega [\xi,h]
    \big).
\end{equation}
The QSDE (\ref{dX}) is driven by a vector $W:= (W_k)_{1\< k \< m}$ of quantum Wiener processes $W_1, \ldots, W_m$ which  are self-adjoint operators on a boson Fock space \cite{H_1991,P_1992}, modelling the external fields. Denoted by $\Omega := (\omega_{jk})_{1\< j,k\< m} \in \mH_m^+$ is the quantum Ito matrix of $W$:
\begin{equation}
\label{Omega}
    \rd W\rd W^{\rT} = \Omega \rd t.
\end{equation}
The dimension $m$ is also assumed to be even, and the entries of $W$ are linear combinations of the field annihilation $\fa_1, \ldots, \fa_{m/2}$ and creation $\fa_1^{\dagger}, \ldots, \fa_{m/2}^{\dagger}$ operator processes \cite{HP_1984,P_1992}:
$$
    W
    :=
    2
    \begin{bmatrix}
        \Re \fa\\
        \Im \fa
    \end{bmatrix}
    =
    \left(
    \begin{bmatrix}1 & 1\\ -i & i\end{bmatrix}
    \ox I_{m/2}
    \right)
    \begin{bmatrix}\fa \\ \fa^{\#}\end{bmatrix},
    \qquad
    \fa:=
    \begin{bmatrix}
        \fa_1\\
        \vdots\\
        \fa_{m/2}
    \end{bmatrix},
    \qquad
    \fa^{\#}:=
    \begin{bmatrix}
        \fa_1^{\dagger}\\
        \vdots\\
        \fa_{m/2}^{\dagger}
    \end{bmatrix},
$$
with the quantum Ito table
$$
    \rd \begin{bmatrix}\fa \\ \fa^{\#}\end{bmatrix}
    \rd \begin{bmatrix}\fa^{\dagger} & \fa^{\rT}\end{bmatrix}
    :=
    \begin{bmatrix}
        \rd\fa \rd \fa^{\dagger} & \rd \fa \rd \fa^{\rT}\\
        \rd\fa^{\#} \rd \fa^{\dagger} & \rd \fa^{\#} \rd \fa^{\rT}
    \end{bmatrix}
    =
    \left(
    \begin{bmatrix}
        1 & 0\\
        0 & 0
    \end{bmatrix}
    \ox I_{m/2}
    \right)
    \rd t.
$$
Accordingly, the Ito matrix $\Omega$ in (\ref{Omega}) is described by
\begin{equation}
\label{OmegaJ}
    \Omega = \left(\begin{bmatrix}1 & 1\\ -i & i\end{bmatrix} \begin{bmatrix}1 & 0\\ 0 & 0\end{bmatrix}
    \begin{bmatrix}1 & 1\\ -i & i\end{bmatrix}^*\right)\ox I_{m/2}=I_m + iJ,
    \qquad
        J :=
    \bJ \ox I_{m/2}.
\end{equation}
Similarly to the CCR matrix $\Theta$ in (\ref{Theta}) and (\ref{ThetaJ}), the matrix $J:= \Im \Omega \in \mA_m$ specifies the cross-commutations between the forward increments of the quantum Wiener processes $W_1, \ldots, W_m$ in the sense that
$
    [\rd W, \rd W^{\rT}] = 2iJ\rd t
$.
In accordance with the evolution (\ref{dX}), the system variables $X_1(t), \ldots, X_n(t)$ at any given time $t\> 0$  act effectively  on a tensor product Hilbert space $\cH_0\ox \cF_t$, where $\cH_0$ is the initial  complex separable Hilbert space of the system (for the action of the operators $X_1(0), \ldots, X_n(0)$), and $\cF_t$ is the Fock filtration.
The structure of the QSDE (\ref{dX}), specified by (\ref{fg})--(\ref{cD}), comes from the Heisenberg unitary evolution on the system-field  composite space $\cH:=\cH_0 \ox \cF$ described by the quantum stochastic flow
\begin{equation}
\label{uni}
    X(t)
    =
    U(t)^{\dagger} (X(0)\ox \cI_{\cF}) U(t),
\end{equation}
where 
the unitary operator $U(t)$ satisfies the initial condition $U(0) = \cI_{\cH}$ and is governed by a stochastic Schr\"{o}dinger equation
$$
    \rd U(t) = -\Big(i\big(h_0(0)\rd t + h(0)^{\rT} \rd W(t)\big) + \frac{1}{2}h(0)^{\rT}\Omega h(0)\rd t\Big)U(t).
$$
  The output field, which results from the interaction of the system with the input field, can be  represented in a similar form as
\begin{equation}
\label{Y}
    Y(t) = U(t)^{\dagger}(\cI_{\cH_0}\ox W(t))U(t)
\end{equation}
except that, with $U(t)$ depending on the past history of the system-field interaction, the right-hand side of (\ref{Y}) involves the current input field variables $W(t)$, which reflects the innovation nature of the quantum Wiener process supported by the continuous tensor product structure of the Fock space \cite{PS_1972}.
The unitary evolution in (\ref{uni}) and (\ref{Y}) preserves the CCRs (\ref{Theta}) and the commutativity between the system variables and output field variables in time $t\> 0$:
\begin{align*}
    [X(t),X(t)^{\rT}]
    & =
    U(t)^{\dagger}
    ([X(0), X(0)^{\rT}] \ox \cI_{\cF})
    U(t)
     =
    2i\Theta
        U(t)^{\dagger}\cI_{\cH_0 \ox \cF}
        U(t)=
    2i\Theta,\\
        [X(t),Y(t)^{\rT}] & = U(t)^{\dagger}[X(0) \ox \cI_{\cF},\cI_{\cH_0}\ox W(t)^{\rT}]U(t)
     =
    0,
\end{align*}
where the entries of $X(0)$ commute with those of $W(t)$ as operators on different spaces.
More general adapted processes $\xi$, which are functions of the system variables, are governed by QSDEs with the same structure as (\ref{dX})--(\ref{fg}):
\begin{equation}
\label{dzeta}
    \rd \xi= \cL(\xi)\rd t - i[\xi,h^{\rT} ]\rd W.
\end{equation}
This property is closely related to the Ito corrected version of the Leibniz product rule for the superoperator $\cD$ in (\ref{cD}) acting on quantum adapted processes $\xi$ and $\eta$:
\begin{align*}
    (\cD(\xi\eta) -\cD(\xi)\eta + \xi \cD(\eta))\rd t
    & =
    - \sum_{j,k=1}^m \omega_{jk}[\xi,h_j][\eta,h_k] \rd t
    =
    -[\xi,h^{\rT}]\Omega [\eta,h] \rd t  = \rd \xi \rd \eta.
\end{align*}

\section{Weyl quantization of the Hamiltonian and coupling operators}\label{sec:weyl}

For what follows, we assume that the system Hamiltonian $h_0$ and the system-field coupling operators $h_1, \ldots, h_m$ in (\ref{fg}) (as functions of the system variables $X_1, \ldots, X_n$) are  obtained from real-valued functions on $\mR^n$ through the  Weyl quantization \cite{F_1989}:
\begin{equation}
\label{hk}
    h_k := \int_{\mR^n} H_k(u)\cW_u \rd u,
    \qquad
    k = 0,1,\ldots, m,
\end{equation}
where $\cW_u$ is the Weyl operator from (\ref{cW}). The Fourier transforms $H_k: \mR^n \to \mC$ of the original classical functions are Hermitian (that is, $H_k(-u) = \overline{H_k(u)}$ for all $u \in \mR^n$), thus ensuring self-adjointness of  the operators $h_k$ in (\ref{hk}) since $\cW_u= \cW_{-u}^{\dagger}$. We assemble the functions $H_1, \ldots, H_m$ into a vector-valued map $H: \mR^n \to \mC^m$, in terms of which the vector $h$ of coupling operators in (\ref{fg}) is  expressed as
\begin{equation}\label{h}
    h = \int_{\mR^n} H(u)\cW_u\rd u,
    \qquad
    H:=
    \begin{bmatrix}
        H_1\\
        \vdots\\
        H_m
    \end{bmatrix}.
\end{equation}
Due to the unitarity of the Weyl operator $\cW_u$ for any $u \in \mR^n$, the integral in (\ref{hk}) can be understood as a Bochner integral  \cite{Y_1980} in the case when the function $H_k$ is absolutely integrable: $\int_{\mR^n} |H_k(u)|\rd u<+\infty$. However, the Fourier transforms $H_k$ can, in principle, be  generalized functions \cite{V_2002}, in which case, the integration in (\ref{hk}) is endowed with an appropriate  distributional meaning. This includes (but is not limited to) the class of polynomials $h_k$. For example,    suppose the Hamiltonian  is a quadratic function and the coupling operators are linear functions of the system variables:
\begin{align}
\label{bR}
    h_0 & := b^{\rT} X + \frac{1}{2} X^{\rT} R X,\\
\label{N}
    h & := NX,
\end{align}
where $b:= (b_j)_{1\< j\< n} \in \mR^n$, $R:= (r_{jk})_{1\< j,k\< n} \in \mS_n$ and $N \in \mR^{m\x n}$. These energy operators can be  represented in the form (\ref{hk}), (\ref{h}) with
\begin{align}
\label{H0lin}
    H_0(u)
    & =
    \sum_{j=1}^n
    \Big(
        ib_j \d_{u_j}\delta(u) - \frac{1}{2}\sum_{k=1}^n r_{jk}\d_{u_j}\d_{u_k}\delta(u)
    \Big)
     =
    ib^{\rT} \delta'(u)-\frac{1}{2}\Tr (R\delta''(u)),\\
\label{Hlin}
    H(u) & = iN\delta'(u),
\end{align}
where $\delta(\cdot)$ is the $n$-dimensional Dirac delta function   with the gradient $\delta'$ and the Hessian matrix $\delta''$. In this case, the system being considered is an $n$-dimensional  open  quantum  harmonic oscillator \cite{EB_2005,GZ_2004} governed by a linear QSDE
\begin{equation}
\label{dXlin}
    \rd X = (A X + 2\Theta b)\rd t + B\rd W,
\end{equation}
where the matrices of coefficients $A\in \mR^{n\x n}$ and $B\in \mR^{n\x m}$ are computed in terms of the matrices $R$ and $N$ from (\ref{bR}) and (\ref{N}) as
\begin{equation}
\label{AB}
    A:= 2\Theta (R + N^{\rT}JN) = 2\Theta R - \frac{1}{2}BJB^{\rT}\Theta^{-1},
    \qquad
    B:= 2\Theta N^{\rT},
\end{equation}
with the second representation of $A$ being valid if $\det \Theta \ne 0$.
The following lemma employs the Weyl quantization (\ref{hk}) in order to represent the drift vector $f$ and the dispersion matrix $g$ of the general QSDE (\ref{dX}) in a similar form.

\begin{lem}
\label{lem:WeylQSDE}
Suppose the Hamiltonian $h_0$ and the coupling operators $h_1, \ldots, h_m$ are given by (\ref{hk}). Then
the drift vector $f$ and the dispersion matrix $g$ of the QSDE (\ref{dX}) in (\ref{fg}) can also be represented in the Weyl quantization form:
\begin{equation}
\label{FG}
  f = \int_{\mR^n} F(u)\cW_u\rd u,
  \qquad
  g = \int_{\mR^n} G(u)\cW_u \rd u.
\end{equation}
Here, $F: \mR^n\to \mC^n$ and $G:\mR^n\to \mC^{n\x m}$ are Hermitian functions which are computed in terms of the Fourier transforms $H_0$ and $H$ from (\ref{hk}) and  (\ref{h}) and an auxiliary function $K: \mR^n \x \mR^n \to \mR^{m\x m}$ as
\begin{align}
\label{F}
    F(u) & := 2i \Theta
    \Big(
        H_0(u)u
        +
      \int_{\mR^n}
         v
         H(v)^{\rT}
         K(u,v)
         H(u-v)  \rd v
    \Big),\\
\label{K}
    K(u,v)
    & :=
    \Im
    \big(
    \re^{iu^{\rT}\Theta v}\Omega
    \big)
    =
    \sin(u^{\rT}\Theta v) I_m + \cos(u^{\rT}\Theta v) J,\\
\label{G}
    G(u) & := 2i\Theta uH(u)^{\rT},
\end{align}
where $\Theta$ is the CCR matrix of the system variables in (\ref{Theta}), and $\Omega$ is the Ito matrix of the quantum Wiener process from (\ref{OmegaJ}). \hfill$\square$
\end{lem}
\begin{proof}
Associated with the Weyl operator $\cW_u$ in (\ref{cW}) is a unitary similarity transformation $\cE_u$ which acts on an operator $\xi$ on the Hilbert space $\cH$ as
\begin{equation}
\label{cE}
  \cE_u(\xi) := \cW_u \xi\cW_{-u} = \re^{i\ad_{u^{\rT} X}}(\xi),
\end{equation}
where 
use is made of a well-known identity for 
operator exponentials  \cite{M_1998,W_1967}. 
The commutator with the Weyl operator can be represented in terms of $\cE_u$ as
\begin{equation}
\label{cExi}
    [\cW_u, \xi]
    =
    (\cE_u(\xi) - \xi) \cW_u.
\end{equation}
Since the CCRs (\ref{Theta})  imply that
$
    [u^{\rT}X,X] = -[X,u^{\rT}X] =-[X,X^{\rT}]u = -2i\Theta u
$,
the entrywise application of the superoperator $\cE_u$ in (\ref{cE}) to the vector $X$ of system variables leads to
\begin{equation}
\label{cEuX}
    \cE_u(X)
    =
    X + 2\Theta u,
\end{equation}
and hence, in view of (\ref{cExi}),
\begin{equation}
\label{XX}
    [\cW_u, X]
    =
    (\cE_u(X)    - X)\cW_u = 2\Theta u \cW_u.
\end{equation}
The identity (\ref{cEuX}) is closely related to the property that $\cW_v$ is an eigenoperator of the superoperator $\cE_u$ with the eigenvalue $\re^{-2iu^{\rT}\Theta v}$ for any $u,v \in \mR^n$:
\begin{equation}
\label{eigop}
    \re^{iv^{\rT}\cE_u(X)}=  \cE_u(\cW_v) = \re^{2iv^{\rT}\Theta u}\cW_v,
\end{equation}
where  the first equality follows from $\cE_u$ being a similarity transformation, while the second equality is obtained from the Weyl CCRs in (\ref{CCR}) and the antisymmetry of the matrix $\Theta$. By combining (\ref{h}) and the bilinearity of the commutator with (\ref{XX}), it follows that the dispersion matrix $g$ in (\ref{fg}) takes the form
\begin{equation}
\label{g}
    g
    =
    -i
    \int_{\mR^n}
    [X,\cW_u]
    H(u)^{\rT}
    \rd u\\
    =
    2i
    \Theta
    \int_{\mR^n}
    u
        H(u)^{\rT}
            \cW_u
    \rd u,
\end{equation}
which establishes the second representation in (\ref{FG}), where $G$ is given by  (\ref{G}). The function $G$ inherits  the Hermitian property ($G(-u) = \overline{G(u)}$  for all $u \in \mR^n$) from $H$. We will now obtain the first equality in (\ref{FG}).
To this end, the term $i[h_0, X]$ of the drift vector $f$ in (\ref{fg}), associated with the internal dynamics (which the system would have in isolation from its environment), can be  represented as
\begin{equation}
\label{h0X}
    i[h_0, X]
    =
    i
    \int_{\mR^n}
    H_0(u)
    [\cW_u, X]
    \rd u
    =
    2i
    \Theta
    \int_{\mR^n}
        H_0(u)
        u
            \cW_u
    \rd u.
\end{equation}
In order to compute the GKSL decoherence term $\cD(X)$ of the drift vector $f$ according to (\ref{cD}), a combination of (\ref{h}) with (\ref{g}) leads to
\begin{align}
\nonumber
    \sum_{j,k=1}^m
    \omega_{jk}
    [h_j, X]h_k
    & =
    -[X, h^{\rT}]\Omega h
=
    -ig\Omega h\\
\nonumber
    & =
    2
    \Theta
    \int_{\mR^n}
    u
        H(u)^{\rT}
            \cW_u
    \rd u
    \Omega
    \int_{\mR^n}
    H(v)
    \cW_v
    \rd v
     =
    2
    \Theta
    \int_{\mR^n\x \mR^n}
    u
        H(u)^{\rT}
        \Omega
        H(v)
            \cW_u
            \cW_v
    \rd u\rd v\\
\label{hXh1}
    & =
    2
    \Theta
    \int_{\mR^n\x \mR^n}
    u
        H(u)^{\rT}
        \Omega
        H(v)
        \re^{-iu^{\rT}\Theta v}
            \cW_{u+v}
    \rd u\rd v
     =
    2
    \Theta
    \int_{\mR^n}
    Q(u)
    \cW_u
    \rd u,
\end{align}
where  $Q: \mR^n\to \mC^n$ is an auxiliary function defined by
\begin{equation}
  \label{Q}
    Q(u)  :=
    \int_{\mR^n}
    v H(v)^{\rT}\Omega H(u-v) \re^{iu^{\rT}\Theta v}\rd v.
\end{equation}
In (\ref{hXh1}), the Weyl CCRs (\ref{CCR}) are combined with a standard measure-preserving transformation of integration variables $(u,v)\mapsto (w,u)$ in convolutions, with $w:= u+v$, along with the relations $u^{\rT}\Theta v = u^{\rT}\Theta (w-u) =  u^{\rT}\Theta w= - w^{\rT}\Theta u$ following from the antisymmetry of the CCR matrix $\Theta$. Due to self-adjointness of the system variables, the superoperator $\cD$ in (\ref{cD}) can be evaluated at the vector $X$ by taking the operator real part of (\ref{hXh1}) as
\begin{equation}
\label{cDX}
    \cD(X)
     =
    -\Re\big( [X,h^{\rT}]\Omega h\big)
      =
     \Theta
     \Big(
    \int_{\mR^n}
    Q(u)
    \cW_u
    \rd u
    +
    \Big(
    \int_{\mR^n}
    Q(u)
    \cW_u
    \rd u
    \Big)^{\#}
 \Big)
  =
 \Theta
    \int_{\mR^n}
    \big(
        Q(u) + \overline{Q(-u)}
    \big)
    \cW_u
    \rd u.
\end{equation}
The Hermitian property of the function $H$ implies that the function $Q$ in (\ref{Q}) satisfies  the identity
\begin{align*}
    \overline{Q(-u)}
    & =
    \int_{\mR^n}
    v H(-v)^{\rT}\overline{\Omega} H(u+v) \re^{iu^{\rT}\Theta v}\rd v
     =
    -
    \int_{\mR^n}
    v H(v)^{\rT}\overline{\Omega} H(u-v) \re^{-iu^{\rT}\Theta v}\rd v,
\end{align*}
and hence,
\begin{equation}
\label{QQ}
    Q(u) + \overline{Q(-u)}
     =
    \int_{\mR^n}
    v H(v)^{\rT}
    \big(
    \re^{iu^{\rT}\Theta v}\Omega
    -
    \re^{-iu^{\rT}\Theta v}
    \overline{\Omega}
    \big)H(u-v) \rd v
     =
    2i
    \int_{\mR^n}
    v H(v)^{\rT}
    K(u,v)
    H(u-v) \rd v,
\end{equation}
where the matrix-valued function $K$ is given by (\ref{K}). 
It now remains to substitute (\ref{QQ}) into (\ref{cDX}) and assemble the resulting decoherence  term   and the internal dynamics term from (\ref{h0X})  into the drift vector $f$ in (\ref{fg}):
$$
    f
     =
    \Theta
    \int_{\mR^n}
    \left(2iH_0(u)u + Q(u) + \overline{Q(-u)}\right)
    \cW_u
    \rd u
     =
    2i
    \Theta
    \int_{\mR^n}
    \Big(
        H_0(u)u
        +
    \int_{\mR^n}
       v
       H(v)^{\rT}
       K(u,v)
       H(u-v)  \rd v
    \Big)
    \cW_u
    \rd u.
$$
 This establishes the first of the equalities in (\ref{FG}), where the function $F$ is given by (\ref{F}) and inherits the Hermitian property from $H_0$ and $H$ in view of the relation $K(-u,v)^{\rT} = -K(u,v)$, thus completing the proof of the lemma.
\end{proof}

Lemma~\ref{lem:WeylQSDE} allows the right-hand side of the QSDE (\ref{dX}) to be decomposed into a ``linear combination'' of  the Weyl operators $\cW_u$, which depend on time through $X$ and play the role of spatial harmonics with different ``wavevectors'' $u \in \mR^n$:
\begin{equation}
\label{dXWeyl}
    \rd X = \int_{\mR^n}(F(u)\rd t + G(u)\rd W)\cW_u\rd u.
\end{equation}
The coefficients of this combination are driven by the quantum Wiener process $W$. We have also used the commutativity between adapted processes and future-pointing  Ito increments of $W$.

\section{Classical limit of the governing QSDE}\label{sec:classlim}

The QSDE (\ref{dX}), whose drift and dispersion are computed in Lemma~\ref{lem:WeylQSDE}, can be related to its classical counterpart by taking into account the reduced Planck constant $\hslash$ as a small parameter. To this end, let the CCR matrix $\Theta$ in (\ref{Theta}) and the matrix $J$ in (\ref{OmegaJ}) be given by
\begin{equation}
\label{Planck}
    \Theta = \frac{\hslash}{2}\Xi ,
    \qquad
    J = \frac{\hslash}{2} \Ups,
\end{equation}
where $\Xi \in \mA_n$ and $\Ups \in \mA_m$ are fixed symplectic structure matrices (for example, $\Xi  = \bJ\ox I_{n/2}$ in accordance with (\ref{ThetaJ})). The internal dynamics, decoherence and dispersion terms in (\ref{fg}) and (\ref{cD}) are appropriately rescaled:
\begin{equation}
\label{dXh}
    \rd X =
    f_{\hslash}
    \rd t
    +
    g_{\hslash}\rd W,
    \qquad
    f_{\hslash}
    :=
    \frac{1}{\hslash}
    \Big(
        i[h_0, X] + \frac{1}{\hslash}\cD(X)
    \Big),
    \qquad
    g_{\hslash}
    :=
    -
    \frac{i}{\hslash}
    [X,h^{\rT}].
\end{equation}
The scaling of the decoherence superoperator $\cD$ comes from its quadratic dependence on the coupling operators $h_1, \ldots, h_m$.
By letting $\hslash\to 0$ in (\ref{Planck}), the drift vector $f_{\hslash}$ and the dispersion matrix $g_{\hslash}$ of the QSDE  (\ref{dXh}) have formal  classical limits
\begin{equation}
\label{f0g0}
    f_0 := \int_{\mR^n} F_0(u)\cW_u\rd u,
  \qquad
  g_0 := \int_{\mR^n} G_0(u)\cW_u\rd u
\end{equation}
which are the inverse Fourier transforms of the corresponding limits of the appropriately rescaled functions $F$ and $G$ from (\ref{F}) and (\ref{G}):
\begin{equation}
\label{F0_G0}
    F_0(u)
     :=
    i\Xi
    \Big(
        H_0(u)u
        +
        \frac{1}{2}
      \int_{\mR^n}
         v
         H(v)^{\rT}
         \left(
             u^{\rT}\Xi  v I_m + \Ups
         \right)
         H(u-v)\rd v
    \Big),
    \qquad
    G_0(u)  := i\Xi  uH(u)^{\rT},
\end{equation}
with $\cW_u$ in (\ref{f0g0}) being the usual exponential function $\re^{iu^{\rT}X}$.
The functions $F_0$ and $G_0$ are the Fourier transforms of the functions
\begin{equation}
\label{fg0}
    f_0
    =
    \Xi  \Big(h_0' + \frac{1}{2}h'^{\rT}\Ups h -\frac{1}{2}\div (h'^{\rT}h' \Xi)\Big)
    =
    \Xi  \Big(h_0' + \frac{1}{2}h'^{\rT}\Ups h\Big) +\frac{1}{2}\div (g_0g_0^{\rT}),
    \qquad
    g_0=\Xi  h'^{\rT},
\end{equation}
where $h_0': \mR^n \to \mR^n$ is the gradient of the classical Hamiltonian $h_0$, and $h': \mR^n \to \mR^{m\x n}$ is the Jacobian matrix of the vector $h$ of classical coupling functions $h_1, \ldots, h_m$ whose Weyl quantization is used in Section~\ref{sec:weyl}. Therefore, the formal classical limit of the QSDE (\ref{dXh}) is an Ito SDE
\begin{equation}
\label{dX0}
    \rd X = f_0(X)\rd t + g_0(X) \rd W
\end{equation}
for a Markov diffusion process $X$
in $\mR^n$ driven by an $m$-dimensional  standard Wiener process with the identity Ito matrix.
If the noncommutativity of the quantum Wiener process $W$ in (\ref{dXh}) is made vanish faster than that of the system variables in the sense that $J= o(\hslash)$ as $\hslash\to 0$, then $\Ups=0$ and the term $\frac{1}{2}h'^{\rT}\Ups h$ disappears from the drift vector $f_0$ in (\ref{fg0}). In this case,  the limit SDE (\ref{dX0}) takes the form
\begin{equation}
\label{can}
    \rd X = \Xi  \Big(\Big(h_0' -\frac{1}{2}\div (h'^{\rT}h' \Xi)\Big)\rd t + h'^{\rT}\rd W\Big).
\end{equation}
This describes a classical stochastic Hamiltonian system with a canonical flow \cite{G_1999,G_2015} which commutes with the Poisson bracket $\{\varphi,\psi\}:= \varphi'^{\rT} \Xi \psi'$  in the sense that
\begin{equation}
\label{can0}
    \rd \{\varphi,\psi\} = \{\rd \varphi,\psi\} + \{\varphi,\rd \psi\} + \{\rd \varphi,\rd \psi\}
\end{equation}
for any smooth functions $\varphi$ and $\psi$ on $\mR^n$. Here, $\{\varphi,\psi\}$ on the left-hand side and $\varphi$, $\psi$ on the right-hand side are evaluated  at $X$ as a time-varying random function of the initial value  $x:= X_0\in \mR^n$. Also,  the differential operators act over $x$, and $\rd \varphi 
=  \cG(\varphi)\rd t + \varphi'^{\rT}g_0\rd W$, where $\cG$ is the infinitesimal generator of $X$ which maps $\varphi$ 
to  $\cG(\varphi):= f_0^{\rT}\varphi'  + \frac{1}{2}\bra g_0g_0^{\rT}, \varphi''\ket_{\rF}$. Indeed, the columns of the dispersion matrix $g_0$ in (\ref{fg0}) are Hamiltonian (and hence, divergenceless) vector fields $\Xi h_1', \ldots, \Xi h_m'$ in $\mR^n$. Therefore, $\div(g_0g_0^{\rT}) = -\Xi\div (h'^{\rT}h' \Xi)=\Xi \sum_{k=1}^m h_k'' \Xi h_k'$, which implies the canonicity of the SDE (\ref{can}) in view of the results of \cite{G_1999,G_2015}. However, the Weyl quantization framework, employed in the present paper, allows the canonical property to be established directly in the spatial frequency domain. More precisely, it suffices to verify (\ref{can0}) for the exponential functions $\varphi(x):= \re^{iu^{\rT}x}$ and $\psi(x):= \re^{iv^{\rT}x}$ with arbitrary   $u,v\in \mR^n$:
\begin{equation}
\label{WW}
    \rd \{\cW_u, \cW_v\} = \{\rd \cW_u, \cW_v\} + \{\cW_u, \rd \cW_v\} + \{\rd \cW_u, \rd \cW_v\}.
\end{equation}
In accordance with the the Doleans-Dade stochastic exponential \cite{DD_1970,LS_2001}, the Ito differential of $\cW_u$ can be computed
as
\begin{equation}
\label{stoc}
    \rd \cW_u
     =
    u^{\rT}
    \Big(\Big(if_0 - \frac{1}{2}g_0g_0^{\rT}u\Big)\rd t + ig_0\rd W\Big) \cW_u
     =
    u^{\rT}
    \int_{\mR^n}
    \Big(
    \Big(
        iF_0(s)
        -
        \frac{1}{2}
        C(s)
        u
    \Big)
    \rd t
    +
    iG_0(s)
    \rd W
    \Big)
    \cW_{s+u}
    \rd s.
\end{equation}
Here, use is made of the quadratic variation  $(u^{\rT}\rd X)^2 = u^{\rT}g_0g_0^{\rT}u\rd t =
|g_0^{\rT}u|^2\rd t$ of the process $u^{\rT}X$ together with (\ref{f0g0}), and the convolution
\begin{equation}
\label{C}
    C(s) :=         (G_0* G_0^{\rT})(s)
    =
    \int_{\mR^n}
        G_0(r)G_0(s-r)^{\rT}
        \rd r
        =
        \Xi
        \int_{\mR^n}
        rH(r)^{\rT}
        H(s-r)
        (s-r)^{\rT}
        \rd r
        \Xi
\end{equation}
is the Fourier transform of the diffusion matrix map $g_0g_0^{\rT}: \mR^n \to \mS_n^+$ for the SDE (\ref{can}) in view of (\ref{F0_G0}). A combination of the Poisson bracket
$
    \{\re^{iu^{\rT}x}, \re^{iv^{\rT}x}\} = -u^{\rT}\Xi v \re^{i(u+v)^{\rT}x}
$ with  (\ref{stoc}) leads to
\begin{align}
\nonumber
    \rd \{\cW_u, \cW_v\}
    & =
    -u^{\rT}\Xi v
    \rd \cW_{u+v}\\
\label{dWW0}
    & =
    -u^{\rT}\Xi v
            (u+v)^{\rT}
    \int_{\mR^n}
    \Big(
    \Big(
        iF_0(s)
        -
        \frac{1}{2}
        C(s)
        (u+v)
    \Big)
    \rd t
    +
    iG_0(s)
    \rd W
    \Big)
    \cW_{s+u+v}
    \rd s,\\
\nonumber
    \{\rd \cW_u, \cW_v\}
    & =
    u^{\rT}
    \int_{\mR^n}
    \Big(
    \Big(
        iF_0(s)
        -
        \frac{1}{2}
        C(s)
        u
    \Big)
    \rd t
    +
    iG_0(s)
    \rd W
    \Big)
    \{\cW_{s+u}, \cW_v\}
    \rd s\\
\label{dWW1}
    & =
    -u^{\rT}
    \int_{\mR^n}
    \Big(
    \Big(
        iF_0(s)
        -
        \frac{1}{2}
        C(s)
        u
    \Big)
    \rd t
    +
    iG_0(s)
    \rd W
    \Big)
    (s+u)^{\rT}\Xi v
    \cW_{s+u+v}
    \rd s,\\
\nonumber
    \{\cW_u, \rd \cW_v\}
    & =
    v^{\rT}
    \int_{\mR^n}
    \Big(
    \Big(
        iF_0(s)
        -
        \frac{1}{2}
        C(s)
        v
    \Big)
    \rd t
    +
    iG_0(s)
    \rd W
    \Big)
    \{\cW_u,  \cW_{s+v}\}
    \rd s\\
\label{dWW2}
    & =
    -v^{\rT}
    \int_{\mR^n}
    \Big(
    \Big(
        iF_0(s)
        -
        \frac{1}{2}
        C(s)
        v
    \Big)
    \rd t
    +
    iG_0(s)
    \rd W
    \Big)
    u^{\rT}\Xi (s+v)
    \cW_{s+u+v}
    \rd s,\\
\nonumber
    \{\rd \cW_u, \rd \cW_v\}
    & =
    -
    u^{\rT}
    \int_{\mR^n\x \mR^n}
    G_0(r)G_0(s)^{\rT}
    \{
        \cW_{r+u},
        \cW_{s+v}
    \}
    \rd r \rd s\,
    v\rd t\\
\nonumber
    & =
    u^{\rT}
    \int_{\mR^n\x \mR^n}
    G_0(r)G_0(s)^{\rT}
    (r+u)^{\rT}
    \Xi
    (s+v)
    \cW_{r+s+u+v}
    \rd r\rd s\,
    v\rd t\\
\label{dWW3}
    & =
    u^{\rT}
    \int_{\mR^n}
    \Xi
    \Big(
    \int_{\mR^n}
    r
    H(r)^{\rT}
    H(s-r)
    (s-r)^{\rT}
    \Xi
    (r+u)^{\rT}
    \Xi
    (s-r+v)
    \rd r
    \Big)
    \cW_{s+u+v}\rd s\,
    v\rd t.
\end{align}
By assembling the drift and diffusion terms in (\ref{dWW0})--(\ref{dWW3}) and considering the Fourier coefficients, it follows that (\ref{WW}) is equivalent to the fulfillment of the relations
\begin{align}
\label{can1}
    0
     =&
    (s^{\rT} \Xi v u^{\rT}
    + u^{\rT} \Xi s v^{\rT})G_0(s), \\
\nonumber
    0
     = &
    i(s^{\rT} \Xi v u^{\rT}
    + u^{\rT} \Xi s v^{\rT})F_0(s)
     +
    \frac{1}{2}
    \Bra
        uv^{\rT} \Xi su^{\rT} + vs^{\rT} \Xi uv^{\rT}+2uu^{\rT} \Xi vv^{\rT},
        C(s)
    \Ket_{\rF}\\
\label{can2}
    & -
    u^{\rT}
    \Xi
    \int_{\mR^n}
    r
    H(r)^{\rT}H(s-r)
    (s-r)^{\rT}
    \Xi
    v\,
    (r+u)^{\rT}
    \Xi
    (s-r+v)
    \rd r
\end{align}
for all $s,u,v\in \mR^n$, where use is also made of the symmetry of the matrix $C(s)$ in (\ref{C}). Now, the functions $F_0$, $G_0$ in (\ref{F0_G0}) indeed satisfy (\ref{can1}), (\ref{can2}) because $(s^{\rT} \Xi v u^{\rT}
    + u^{\rT} \Xi s v^{\rT})\Xi s = (s^{\rT} \Xi v + v^{\rT} \Xi s)u^{\rT} \Xi s = 0$ in view of the antisymmetry of $\Xi$, and
\begin{align*}
    -(s^{\rT} \Xi v u^{\rT}
    &+ u^{\rT} \Xi s v^{\rT})
        \Xi
      \int_{\mR^n}
         r
         H(r)^{\rT}
         H(s-r)
         s^{\rT}\Xi  r\rd r\\
     & +
    \Bra
        uv^{\rT} \Xi su^{\rT} + vs^{\rT} \Xi uv^{\rT}+2uu^{\rT} \Xi vv^{\rT},\,
                \Xi
        \int_{\mR^n}
        rH(r)^{\rT}
        H(s-r)
        (s-r)^{\rT}
        \rd r
        \Xi
    \Ket_{\rF}\\
    & -
    2u^{\rT}
    \Xi
    \int_{\mR^n}
    r
    H(r)^{\rT}H(s-r)
    (s-r)^{\rT}
    \Xi
    v\,
    (r+u)^{\rT}
    \Xi
    (s-r+v)
    \rd r
    =
    \int_{\mR^n}
    \phi(r,s,u,v)
    H(r)^{\rT}H(s-r)
    \rd r = 0.
\end{align*}
The last integral vanishes since $H(r)^{\rT}H(s-r)$ is invariant, while the following function is antisymmetric, under the transformation $r\mapsto s-r$:
\begin{align*}
    \phi(r,s,u,v)
    := &
    -(s^{\rT}\Xi vu^{\rT}+u^{\rT}\Xi sv^{\rT})\Xi rs^{\rT}\Xi r\\
        & +
        (s-r)^{\rT}\Xi (uv^{\rT}\Xi su^{\rT}+ vs^{\rT}\Xi uv^{\rT}-2vv^{\rT}\Xi uu^{\rT})
        \Xi r\\
    & -
    2u^{\rT}\Xi r(s-r)^{\rT}\Xi v(r+u)^{\rT}\Xi (s-r+v)
    =
    -\phi(s-r,s,u,v).
\end{align*}

\section{Quasi-characteristic function and 
generalized moments}\label{sec:quasi}

We now return to the quantum system. Its averaged behaviour can be described in terms of the quasi-characteristic function (QCF) \cite{CH_1971,M_1949} defined by
\begin{equation}
\label{char}
    \Phi(t,u):= \bE \cW_u(t) = \overline{\Phi(t,-u)},
    \qquad
    t\> 0, \ u \in \mR^n,
\end{equation}
where $\cW_u(t)$ is the Weyl operator (\ref{cW}) at time $t$. We assume that the quantum expectation is over the density operator $\rho := \varpi \ox \ups$, where $\varpi$ is the initial quantum state of the system, and $\ups$ is the vacuum state of the external bosonic fields \cite{P_1992}.
The QCF $\Phi: \mR_+ \x \mR^n \to \mC$ encodes the mixed moments of the system variables:
\begin{equation}
\label{mom}
    \bE(X(t)^{\alpha})
    =
    (-i\d_u)^{\alpha}
    \bE     \rprod_{k=1}^n
    \re^{iu_k X_k(t)}\Big|_{u=0}
    =
    (-i\d_u)^{\alpha}
    \Big(
        \Phi(t,u) \re^{-\frac{i}{2}u^{\rT}\wt{\Theta}u}
    \Big)
    \Big|_{u = 0}
\end{equation}
for any $n$-index  $\alpha:= (\alpha_k)_{1\< k\< n}\in \mZ_+^n$ such that $\Phi$ is $|\alpha|$ times continuously differentiable  with respect to $u:= (u_k)_{1\< k\< n}\in \mR^n$. Here, $\wt{\Theta}:= (\wt{\theta}_{jk})_{1\< j,k\< n} \in \mS_n$ is an auxiliary matrix which is defined by
\begin{equation}
\label{thetat}
    \wt{\theta}_{jk}
    :=
    \left\{
        \begin{matrix}
            0 & {\rm if} & j = k\\
            \theta_{jk} & {\rm if} & j < k\\
            \theta_{kj} & {\rm if} & j > k
        \end{matrix}
    \right.
\end{equation}
and inherits the upper off-diagonal entries of the CCR matrix $\Theta$.
Since $\wt{\Theta}$ has zero main diagonal (and hence, $\Tr \wt{\Theta} = 0$), this matrix is indefinite.
For example, if $\Theta$ is given by (\ref{ThetaJ}), then
\begin{equation}
\label{exThetat}
    \wt{\Theta}
    =
    \frac{1}{2}
    \begin{bmatrix}
        0 & 1\\
        1 & 0
    \end{bmatrix}\ox I_{n/2}
\end{equation}
is an indefinite matrix with the eigenvalues $\pm \frac{1}{2}$ of multiplicity $\frac{n}{2}$.
The relation (\ref{mom}) is obtained by averaging the identity
\begin{equation}
\label{Baker}
    \cW_u
      =
    \rprod_{k=1}^n
    \Big(
    \cW_{u_k e_k}
    \re^{i\sum_{j=1}^{k-1}u_j e_j^{\rT}\Theta e_ku_k }
    \Big)
        =
    \re^{i\sum_{1\< j< k \< n}\theta_{jk} u_ju_k}
    \rprod_{k=1}^n
    \cW_{u_ke_k}
    =
    \re^{\frac{i}{2}u^{\rT}\wt{\Theta}u}
    \rprod_{k=1}^n
    \re^{iu_k X_k}
\end{equation}
which is established by repeatedly using the Weyl CCRs (\ref{CCR}) in combination with the bilinearity of the commutator, where
$e_k$ denotes the $k$th standard basis vector in $\mR^n$, so that $\cW_{u_k e_k} = \re^{iu_k X_k}$.
Alternatively, (\ref{Baker}) can be obtained by repeated application of the Baker-Campbell-Hausdorff formula $\re^{\xi+\eta} = \re^{\xi} \re^{\eta}
\re^{-\frac{1}{2}[\xi,\eta]} $ for operators $\xi$ and $\eta$ which commute with their commutator \cite
{GZ_2004,W_1967}.
In what follows, we will also employ the Wigner quasi-probability density function (QPDF) $\mho: \mR_+\x \mR^n \to \mR$ which is the Fourier transform of the QCF $\Phi$ in (\ref{char}):
\begin{equation}
\label{Wig}
    \mho(t,x):= (2\pi)^{-n}\int_{\mR^n} \Phi(t,u)\re^{-iu^{\rT}x}\rd u,
    \qquad
    t\> 0, \ x \in \mR^n.
\end{equation}
The function $\mho$ is real-valued due to the Hermitian property of $\Phi$ and satisfies the normalization condition
\begin{equation}
\label{one}
    \int_{\mR^n}\mho(t,x)\rd x = \Phi(t,0)=1,
\end{equation}
thus resembling a PDF in $\mR^n$ despite not necessarily being nonnegative everywhere \cite{H_1974}. If $\mho$ is negative on a subset of nonzero  Lebesgue measure in $\mR^n$, the QCF $\Phi$ is not positive definite in view of the Bochner-Khinchin  theorem.
Although the exponential moments in (\ref{char}) are related to the QPDF $\mho$ in (\ref{Wig}) by the inverse Fourier transform
\begin{equation}
\label{Phimho}
    \Phi(t,u)
    =
    \int_{\mR^n} \mho(t,x)\re^{iu^{\rT}x}\rd x,
\end{equation}
similarly to the classical case, the mixed moments of the system variables in (\ref{mom}) are expressed in terms of $\mho$ in a more complicated fashion as
\begin{equation}
\label{momWig}
    \bE(X(t)^{\alpha})
    =
    (-i\d_u)^{\alpha}
    \Big(
        \re^{-\frac{i}{2}u^{\rT}\wt{\Theta}u}
        \int_{\mR^n} \mho(t,x)\re^{iu^{\rT}x}\rd x
    \Big)
    \Big|_{u = 0}
    =
    \int_{\mR^n}
    \Psi_{\alpha}(x)
    \mho(t,x)
    \rd x.
\end{equation}
Here, for any $n$-index $\alpha\in \mZ_+^n$, the function $\Psi_{\alpha}: \mR^n\to \mC$ is a polynomial of degree $|\alpha|$ defined by
\begin{equation}
\label{Psi}
    \Psi_{\alpha}(x)
    :=
(-i\d_u)^{\alpha}
    \re^{i(u^{\rT}x - \frac{1}{2}u^{\rT}\wt{\Theta}u)}\Big|_{u=0}
    =
\d_u^{\alpha}
    \re^{u^{\rT}x + \frac{i}{2}u^{\rT}\wt{\Theta}u}\Big|_{u=0}
    =
    (-1)^{|\alpha|}
    \re^{-\frac{i}{2}z^{\rT}\wt{\Theta}z}
    \d_z^{\alpha}
    \re^{\frac{i}{2}z^{\rT}\wt{\Theta}z}
    \Big|_{z=i\wt{\Theta}^{-1} x},
\end{equation}
provided $\det \wt{\Theta}\ne 0$,
with $\re^{- \frac{i}{2}u^{\rT}\wt{\Theta}u}$
playing the role of a quantum correcting factor in comparison with moments of classical random variables. In the classical limit (when the system variables $X_1(t), \ldots, X_n(t)$ commute with each other, and hence, $\Theta=\wt{\Theta} = 0$ in view of (\ref{thetat}), and the function $\mho(t,\cdot)$ in (\ref{Wig}) becomes their usual joint PDF), the polynomial $\Psi_{\alpha}(x)$ reduces to the monomial  $x^{\alpha}$ in accordance with (\ref{momWig}). In the noncommutative case being considered, the polynomials $\Psi_{\alpha}$ in (\ref{Psi}) have the generating function
\begin{equation}
\label{polgen}
    \sum_{\alpha \in \mZ_+^n}
    \Psi_{\alpha}(x)\frac{u^{\alpha}}{\alpha !}
    =
    \re^{u^{\rT}x + \frac{i}{2}u^{\rT}\wt{\Theta}u}
\end{equation}
and can be regarded as a quantum mechanical modification of the multivariate Hermite polynomials. Although (\ref{polgen}) resembles the generating function representation of the standard Hermite polynomials  \cite{M_1997}, the qualitative difference is that  the matrix $\wt{\Theta}$ is indefinite (see also the example (\ref{exThetat})). The polynomials $\Psi_{\alpha}$ play a role for more general moments of the system variables. More precisely, \begin{align}
\nonumber
    \d_u^{\alpha} \cW_u
    & =
    \d_v^{\alpha} \cW_{u+v}\Big|_{v=0}
    =
    \d_v^{\alpha}
    \big(
        \re^{-iu^{\rT}\Theta v}
        \cW_v
    \big)\Big|_{v=0}\cW_u
=
    \d_v^{\alpha}
    \Big(
        \re^{iv^{\rT}\Theta u + \frac{i}{2}v^{\rT} \wt{\Theta} v}
        \rprod_{k=1}^n
        \re^{iv_k X_k}
    \Big)
    \Big|_{v=0}
    \cW_u\\
\nonumber
    & =
    \sum_{\beta\in \mZ_+^n:\ \beta \< \alpha}
    \begin{pmatrix}
        \alpha\\
        \beta
    \end{pmatrix}
    \d_v^{\alpha-\beta}
    \re^{iv^{\rT}\Theta u + \frac{i}{2}v^{\rT} \wt{\Theta} v}\Big|_{v=0}
    \d_v^{\beta}
            \rprod_{k=1}^n
        \re^{iv_k X_k}\Big|_{v=0}
        \cW_u\\
\label{XXW}
    & =
    \alpha!
    \sum_{\beta\in \mZ_+^n:\ \beta \< \alpha}
    \frac{i^{|\beta|}}{\beta!(\alpha-\beta)!}
    \Psi_{\alpha-\beta}(i\Theta u)
    X^{\beta}
    \cW_u,
\end{align}
which follows from the Weyl CCRs (\ref{CCR}), the factorization (\ref{Baker}) and the Leibniz product  rule combined with (\ref{Psi}), with the multiindex inequality $\beta \< \alpha$ being understood entrywise. The averaging of both parts of (\ref{XXW}) leads to
\begin{equation}
\label{PhiXX}
  \d_u^{\alpha} \Phi
  =
  \alpha!
  \sum_{\beta\in \mZ_+^n:\ \beta \< \alpha}
  \frac{i^{|\beta|}}{\beta!(\alpha-\beta)!}
  \Psi_{\alpha-\beta}(i\Theta u)
  \bE(X^{\beta} \cW_u),
\end{equation}
which is another moment identity for the QCF $\Phi$ involving
quasi-polynomials of the system variables, thus extending (\ref{mom}). In particular, by considering (\ref{XXW}) and (\ref{PhiXX}) for those  $n$ multiindices $\alpha \in \mZ_+^n$ which satisfy $|\alpha|=1$ (and can be represented as $\alpha = e_j$ for $j = 1, \ldots, n$), it follows that
\begin{equation}
\label{XW}
    \d_u \cW_u =
    \d_v \re^{iv^{\rT}(X+\Theta u)}\Big|_{v=0}     \cW_u
    =
    i(X+\Theta u) \cW_u.
\end{equation}
Hence, the gradient $\d_u\Phi$ and the Hessian matrix $\d_u^2 \Phi$ of the QCF $\Phi$ with respect to $u$ satisfy the identities
\begin{equation}
\label{PhiX}
    \d_u \Phi = i(\bE(X\cW_u) + \Phi\Theta u),
    \qquad
    \d_u^2 \Phi = i\Phi \Theta -\bE\big((X+\Theta u)(X+\Theta u)^{\rT}\cW_u\big),
\end{equation}
whereby the mean vector $\mu$ of the system variables and the real part $\Sigma$ of their quantum covariance matrix
\begin{equation}
\label{covX}
    \cov(X)
    =
    \bE(XX^{\rT}) - \mu\mu^{\rT} = \Sigma +i\Theta
\end{equation}
can be represented in terms of the QPDF $\mho$ by the same relations as for classical random variables:
\begin{align}
\label{mu}
    \mu(t) & :=
    \bE X(t)
    =-i\d_u\Phi(t,u)
    =
    \int_{\mR^n}
    x \mho(t,x)\rd x,\\
\label{Sigma}
    \Sigma(t)
    & :=
    \Re \cov(X(t))
    =-\d_u^2\Phi(t,u)-\mu\mu^{\rT}
    =
    \int_{\mR^n}
    xx^{\rT} \mho(t,x)\rd x
    -\mu\mu^{\rT}.
\end{align}
Furthermore, the QCF $\Phi$ in (\ref{char})
can be used for evaluating generalized moments of the system variables involving nonlinear (but not necessarily polynomial) functions in the Weyl quantization form:
\begin{equation}
\label{genmom}
    \bE \int_{\mR^n}\sigma(u) \cW_u(t)\rd u = \int_{\mR^n}\sigma(u) \Phi(t,u)\rd u,
\end{equation}
where $\sigma: \mR^n\to \mC$ is a given function which specifies the moment under consideration (and can be a generalized function as discussed before). For example,  the moments of the form (\ref{genmom}) drive the mean vector of the system variables:
\begin{equation}
\label{mudot}
    \dot{\mu}(t)
    =
    \int_{\mR^n}
    F(u)\Phi(t,u)\rd u,
\end{equation}
which is obtained by averaging the QSDE (\ref{dXWeyl}), with the martingale part $g\rd W$ not contributing to the time derivative of the average.
The ODE (\ref{mudot}) is not algebraically closed since, in general, the mean vector $\mu$ does not specify the QCF $\Phi$ uniquely.

\section{Integro-differential equation for the quasi-characteristic function}\label{sec:evol}

In contrast to the dynamics of the system variables described by the nonlinear QSDE (\ref{dX}),  the time evolution of the QCF $\Phi$ is governed by a linear equation.

\begin{thm}
\label{th:Phidot}
Suppose the system Hamiltonian $h_0$ and the system-field coupling operators $h_1, \ldots, h_m$ are given by (\ref{hk}). Then the QCF $\Phi$ of the system variables, defined by (\ref{char}), satisfies a linear IDE
\begin{equation}
\label{Phidot}
    \d_t\Phi(t,u)
    =
    \int_{\mR^n}
    V(u,s)\Phi(t,u+s)
    \rd s.
\end{equation}
The kernel function $V: \mR^n \x \mR^n\to \mC$  of the integral operator is computed in terms of the functions $F$, $G$ from (\ref{F}), (\ref{G}) as
\begin{align}
\nonumber
  V(u,s)
  :=&
  i\,\sinc(u^{\rT}\Theta s)u^{\rT}F(s)
  -u^{\rT}
    \int_{\mR^n}
    G(v)
    L(u,v,s)
    G(s-v)^{\rT}
    \rd v
    \,
    u\\
\label{V}
    =&
    -2        \sin(u^{\rT}\Theta s)H_0(s)
    -2 \int_{\mR^n}
            \sin(u^{\rT}\Theta v)
            H(v)^{\rT}
            K(u+v,v-s)
        H(s-v) \rd v,
\end{align}
where $H_0$ and $H$  are the Fourier transforms from (\ref{hk}) and  (\ref{h}),
the function $K$ is given by (\ref{K}), and $L: \mR^n\x \mR^n\x \mR^n\to \mR^{m\x m}$ is defined in terms of another auxiliary function $M: \mR^n\x \mR^n\x \mR^n\to  \mC$ by
\begin{align}
\label{Luvs}
    L(u,v,s)
    &:=
        \Re
    \big(
        M(u,v,s-v)
        \re^{is^{\rT}\Theta v}
        \Omega
    \big)
    =
    \Re M(u,v,s-v) \Re(        \re^{is^{\rT}\Theta v}
        \Omega)
        -
    \Im M(u,v,s-v)
    K(s,v),
    \\
\label{Muvw}
    M(u,v,w)
    & :=
    \frac{1}{2}
    \sinc(u^{\rT}\Theta v)\sinc(u^{\rT}\Theta w)
     +
    i\,
     \frac{u^{\rT}\Theta (v-w)(\sinc(u^{\rT}\Theta (v-w))-\sinc(u^{\rT}\Theta (v+w)))}{4u^{\rT}\Theta v\, u^{\rT}\Theta w}.
\end{align}
\hfill$\square$
\end{thm}
\begin{proof}
For any fixed but otherwise arbitrary $u \in \mR^n$,  application of the QSDE (\ref{dzeta}) to the Weyl operator $\cW_u$ leads to
\begin{equation}
\label{de}
    \rd \cW_u
    =
    \cL(\cW_u)
    \rd t
    -i[\cW_u, h^{\rT}] \rd W.
\end{equation}
This QSDE extends (\ref{dX}) in the sense that the terms $f$ and $g$ of the latter QSDE, described by (\ref{fg}), can be recovered from (\ref{de}) as operator-valued coefficients  of the linear parts of the formal power series
$
    \cL(\cW_u) = iu^{\rT}f + o(u)
$ and
$
    -i[\cW_u, h^{\rT}] = iu^{\rT} g + o(u)
$,
where $o(u)$ consists of higher-order monomials of $u$.
The averaging of the QSDE (\ref{de}) (to which the martingale part $-i[\cW_u, h^{\rT}] \rd W$ does not contribute) yields the following IDE for the QCF $\Phi$ in (\ref{char}):
\begin{equation}
\label{PhidotEcL}
  \d_t \Phi(t,u) = \bE \cL(\cW_u(t)).
\end{equation}
Although the drift of the QSDE (\ref{de}) can be computed directly through (\ref{cL}) and (\ref{cD}), we will follow a slightly longer path  based on applying \cite[Lemma 2]{VP_2012b} to the quantum adapted process $iu^{\rT}X$ satisfying the QSDE $\rd (iu^{\rT} X) = iu^{\rT}f\rd t + iu^{\rT} g \rd W$ in view of (\ref{dX}), which leads to the representation
\begin{equation}
\label{cLcA}
    \cL(\cW_u) = \cW_{u/2} \cA_u \cW_{u/2}.
\end{equation}
Here, $\cA_u$ is an adapted quantum process which is defined in terms of the drift vector $f$ and the dispersion matrix $g$ as
\begin{equation}
\label{cA}
    \cA_u
    :=
    iu^{\rT}
    \int_{-\frac{1}{2}}^{\frac{1}{2}}
    \Big(
    \cE_{r u}(f)
    +
    i
    \cE_{ru}(g)
    \Omega
    \int_r^{\frac{1}{2}}
    \cE_{zu}(g)^{\rT}
    \rd z\,  u
    \Big)
    \rd r,
\end{equation}
where $\cE_u$ is the superoperator given by (\ref{cE}), and $\Omega$ is the Ito matrix from (\ref{Omega}).  The representation (\ref{cLcA}) in terms of (\ref{cA})  can also be established by using the general quantum stochastic exponential formulae \cite{H_1996} and relates (\ref{de}) with the following noncommutative analogue of the Doleans-Dade exponential \cite{DD_1970,LS_2001}:
\begin{equation}
\label{DD}
    \rd \cW_u = \cW_{u/2}(\cA_u\rd t + \cB_u \rd W)\cW_{u/2},
    \qquad
    \cB_u
    :=
    iu^{\rT}
    \int_{-\frac{1}{2}}^{\frac{1}{2}}
    \cE_{r u}(g)
    \rd r.
\end{equation}
In addition to its connection with the classical stochastic exponentials,  the representation (\ref{cLcA})--(\ref{cA}) allows the Weyl quantization form of $f$ and $g$ from Lemma~\ref{lem:WeylQSDE} to be combined with the eigenrelation (\ref{eigop}). Indeed, this relation implies that
$$    \int_{-\frac{1}{2}}^{\frac{1}{2}}\cE_{ru}(\cW_v)\rd r
     =
    \int_{-\frac{1}{2}}^{\frac{1}{2}}\re^{-2iru^{\rT}\Theta v}\rd r\, \cW_v
    =
    \sinc(u^{\rT}\Theta v) \cW_v
$$
for any $v \in \mR^n$, and hence, due to the linearity of the superoperator $\cE_u$ and the representation of $f$ in (\ref{FG}), the leftmost integral in (\ref{cA}) can be evaluated as
\begin{equation}
\label{intcEf}
        \int_{-\frac{1}{2}}^{\frac{1}{2}}
    \cE_{ru}(f)
    \rd r
    =
    \int_{\mR^n}
    F(v)
        \int_{-\frac{1}{2}}^{\frac{1}{2}}
    \cE_{ru}(\cW_v)
    \rd r
    \rd v
    =
    \int_{\mR^n}
    \sinc(u^{\rT} \Theta v)
    F(v)
    \cW_v
    \rd v.
\end{equation}
By a similar reasoning, substitution of the Weyl quantization form of $g$ from (\ref{FG}) into (\ref{cA}) leads to
\begin{align}
\nonumber
    \int_{-\frac{1}{2}}^{\frac{1}{2}}
    \cE_{ru}(g)
     \Omega
    \int_r^{\frac{1}{2}}
    \cE_{zu}(g)^{\rT}
    \rd z
    \rd r
    & =
    \int_{-\frac{1}{2}}^{\frac{1}{2}}
    \int_{\mR^n}
    G(v)
    \cE_{ru}(\cW_v)
    \rd v\,
    \Omega
    \int_r^{\frac{1}{2}}
    \int_{\mR^n}
    G(w)^{\rT}
    \cE_{zu}(\cW_w)
    \rd w
    \rd z
    \rd r\\
\nonumber
    & =
    \int_{\mR^n\x \mR^n}
        M(u,v,w)
        \re^{-iv^{\rT}\Theta w}
        G(v)\Omega G(w)^{\rT}
        \cW_{v+w}
    \rd v \rd w\\
\label{LGG}
    & =
    \int_{\mR^n}
    \Big(
    \int_{\mR^n}
        M(u,v,s-v)
        \re^{is^{\rT}\Theta v}
        G(v)\Omega G(s-v)^{\rT}
        \rd v
        \Big)
        \cW_s
    \rd s.
\end{align}
Here,
\begin{equation}
\label{int2}
    \int_{-\frac{1}{2}}^{\frac{1}{2}}\cE_{ru}(\cW_v)
    \int_r^{\frac{1}{2}}
    \cE_{zu}(\cW_w)\rd z\rd r
     =
    \int_{-\frac{1}{2}\< r\< z\< \frac{1}{2}}
    \re^{-2iu^{\rT}\Theta (rv + zw)}
    \rd r\rd z\,
    \cW_v \cW_w
     =
    M(u,v,w)
    \re^{-iv^{\rT}\Theta w}
    \cW_{v+w}
\end{equation}
for all $u,v,w\in \mR^n$, where the function $M$ is given by (\ref{Muvw}) and 
results from computing the rightmost integral in (\ref{int2}) as
\begin{align}
\nonumber
    \int_{-\frac{1}{2}\< r\< z\< \frac{1}{2}}
    \re^{2i(\beta r+ \gamma z)}
    \rd r \rd z
    & =
    \int_{-\frac{1}{2}}^{\frac{1}{2}}
    \re^{2i\beta r}
    \int_{r}^{\frac{1}{2}}
    \re^{2i\gamma z}
    \rd z\,
    \rd r\\
\nonumber
    & =
    \frac{1}{2i\gamma}
    \int_{-\frac{1}{2}}^{\frac{1}{2}}
    \re^{2i\beta r}
    \big(
        \re^{i\gamma} - \re^{2i\gamma r}
    \big)
    \rd r
 =
    \frac{\sinc(\beta)\re^{i\gamma} -\sinc(\beta+\gamma)}{2i\gamma}\\
\label{intsinc}
    & =
    \frac{1}{2}\sinc(\beta)\sinc(\gamma)
    +
    i\,
    \frac{(\beta-\gamma)(\sinc(\beta+\gamma)-\sinc(\beta-\gamma))}{4\beta \gamma},
\end{align}
with $\beta := -u^{\rT}\Theta v$ and $\gamma := -u^{\rT}\Theta w$. If $\beta=0$ or $\gamma=0$,  the right-hand side of (\ref{intsinc}) is defined by continuity. In particular, it is equal to $\frac{1}{2}$ for $\beta = \gamma = 0$ in accordance with the leftmost integral in (\ref{intsinc}) over the triangle of area $\frac{1}{2}$. By assembling (\ref{intcEf}) and (\ref{LGG}) into (\ref{cA}), the process $\cA_u$  takes the form
\begin{equation}
\label{cAweyl}
    \cA_u
    =
    \int_{\mR^n}
    \Big(
    i\,\sinc(u^{\rT}\Theta s)u^{\rT}F(s)
     -
    u^{\rT}
    \cS
    \Big(
    \int_{\mR^n}
            M(u,v,s-v)
        \re^{is^{\rT}\Theta v}
    G(v)
        \Omega
    G(s-v)^{\rT}
    \rd v
    \Big)
    u
    \Big)
        \cW_s
    \rd s,
\end{equation}
where we have used the property $u^{\rT} \Gamma u = u^{\rT} \cS(\Gamma) u$ for any vector $u$ and matrix $\Gamma$. In order to compute the symmetrizer in
(\ref{cAweyl}), we note that
\begin{align}
\nonumber
    \Big(
        \int_{\mR^n}
        M(u,v,s-v)
        \re^{is^{\rT}\Theta v}
        G(v)\Omega G(s-v)^{\rT}
        \rd v
        \Big)^{\rT}
        & =
        \int_{\mR^n}
        M(u,v,s-v)
        \re^{is^{\rT}\Theta v}
        G(s-v)\overline{\Omega} G(v)^{\rT}
        \rd v \\
\nonumber
        & =
        \int_{\mR^n}
        M(u,s-v,v)
        \re^{is^{\rT}\Theta (s-v)}
        G(v)\overline{\Omega} G(s-v)^{\rT}
        \rd v \\
\nonumber
        & =
        \int_{\mR^n}
        \overline{M(u,v,s-v)}
        \re^{-is^{\rT}\Theta v}
        G(v)\overline{\Omega} G(s-v)^{\rT}
        \rd v\\
\nonumber
        & =
        \int_{\mR^n}
        G(v)\,
        \overline{M(u,v,s-v)\re^{is^{\rT}\Theta v} \Omega}\,
        G(s-v)^{\rT}
        \rd v,
\end{align}
where use is made of the Hermitian property $\Omega^{\rT} = \overline{\Omega}$ of the Ito matrix $\Omega$, the identity
$M(u,w,v) = \overline{M(u,v,w)}$ for the function $M$ in (\ref{Muvw}), and the antisymmetry of the CCR matrix $\Theta$. Therefore,
\begin{equation}
\label{symm}
      \cS
    \Big(
    \int_{\mR^n}
            M(u,v,s-v)
        \re^{is^{\rT}\Theta v}
    G(v)
        \Omega
    G(s-v)^{\rT}
    \rd v
    \Big)
    =
        \int_{\mR^n}
        G(v)
        L(u,v,s)
        G(s-v)^{\rT}
        \rd v,
\end{equation}
where the function $L$ is defined in terms of $M$ by (\ref{Luvs}). In view of (\ref{symm}), the integral representation (\ref{cAweyl}) takes the form \begin{equation}
\label{AVW}
  \cA_u =
  \int_{\mR^n}
  V(u,s) \cW_s \rd s.
\end{equation}
Here, the kernel function $V$ is given by (\ref{V}), where the second equality is obtained by using (\ref{F}) and (\ref{G}) as
\begin{align}
\nonumber
  V(u,s)
  :=&
    i\, \sinc(u^{\rT}\Theta s)
    u^{\rT} F(s)
    -u^{\rT}
    \int_{\mR^n}
    G(v)
    L(u,v,s)
    G(s-v)^{\rT}
    \rd v
    \,
    u\\
\nonumber
    =&
    -2
    \sinc(u^{\rT}\Theta s)
    u^{\rT}\Theta
    \Big(
        H_0(s)s
        +
      \int_{\mR^n}
         v
         H(v)^{\rT}
         K(s,v)
         H(s-v)  \rd v
    \Big)\\
\nonumber
    & -
    4u^{\rT}\Theta \int_{\mR^n}vH(v)^{\rT}L(u,v,s)H(s-v)(s-v)^{\rT}\rd v \Theta u\\
\nonumber
    =&
    -2\sin(u^{\rT}\Theta s)H_0(s)
     +
    \int_{\mR^n}
    H(v)^{\rT}
    \big(
        4u^{\rT}\Theta v u^{\rT}\Theta (s-v)L(u,v,s) - 2\sinc(u^{\rT}\Theta s)u^{\rT}\Theta v K(s,v)
    \big)
    H(s-v)
    \rd v
    \\
\nonumber
   = & -2\sin(u^{\rT}\Theta s)H_0(s)
-2
    \int_{\mR^n}
    \sin(u^{\rT}\Theta v)
    H(v)^{\rT}
    K(u+v,v-s)
    H(s-v)
    \rd v,
\end{align}
where we have also employed the identity
$$
\sinc(u^{\rT}\Theta s)u^{\rT}\Theta v K(s,v) -2u^{\rT}\Theta v u^{\rT}\Theta (s-v)L(u,v,s)
=
\sin(u^{\rT}\Theta v)K(u+v,v-s)
$$
which follows from (\ref{K}), (\ref{Luvs}) and (\ref{Muvw}).
By substituting (\ref{AVW}) into (\ref{cLcA}) and using the relation $\cW_{u/2} \cW_s \cW_{u/2} =
\cW_{u+s}$, which holds  for all $u,s\in \mR^n$ in view of the Weyl CCRs (\ref{CCR}), it follows that
\begin{equation}
\label{cLWWW}
    \cL(\cW_u) = \cW_{u/2} \int_{\mR^n} V(u,s)\cW_s\rd s \cW_{u/2} = \int_{\mR^n} V(u,s)\cW_{u+s}\rd s.
\end{equation}
The IDE (\ref{Phidot}) can now be obtained by averaging (\ref{cLWWW}) and using (\ref{PhidotEcL}), which completes the proof of the theorem.
 \end{proof}

If the QSDE (\ref{dX}) were a classical SDE for an $\mR^n$-valued Markov diffusion process $X$ driven by a standard Wiener process $W$, the drift vector $f$ and the dispersion matrix $g$ would be usual functions on $\mR^n$ with the Fourier transforms $F$ and $G$. In the classical case, (\ref{DD}) would reduce to the stochastic exponential 
$$
    \rd \cW_u
    =
    \Big(iu^{\rT}\rd X - \frac{1}{2}|g^{\rT}u|^2\rd t\Big) \cW_u
=
 u^{\rT}\Big(\Big(if - \frac{1}{2}gg^{\rT}u\Big)\rd t + ig\rd W\Big) \cW_u,
$$
which corresponds to the previously discussed classical limit (\ref{stoc}),
and the characteristic function $\Phi$ would satisfy the IDE
$$
    \d_t \Phi(t,u)
    =
    u^{\rT}
    \int_{\mR^n}
        \Big(
                i F(s)
            -
            \frac{1}{2}
            (G*G^{\rT})(s)\,
            u
        \Big)
        \Phi(t,u+s)
        \rd s.
$$
This IDE can also be  obtained through the Fourier transform of the Fokker-Planck-Kolmogorov equation (FPKE) \cite{KS_1991,S_2008} which governs the PDF of the random process $X$ in $\mR^n$, with the convolution $G* G^{\rT}$ 
being the Fourier transform of the diffusion matrix $gg^{\rT}$ for  the classical SDE (cf. (\ref{C})).
In the quantum case, the term $-2        \sin(u^{\rT}\Theta s)H_0(s)$ in (\ref{V}) comes from the internal dynamics of the system and corresponds to the kernel of the Moyal equation \cite[Eq. (7.4)]{M_1949} for the evolution of the QCF for isolated quantum systems. The integral term on the right-hand side of (\ref{V}) describes the contribution from the  system-field interaction, thus making the IDE (\ref{Phidot}) an extension of the Moyal equation to open quantum stochastic systems.

\section{QCF and QPDF dynamics in the case of linear system-field coupling}\label{sec:Hlin}

Consider a particular class \cite{SVP_2014} of the above discussed quantum systems whose coupling operators are linear functions of the system variables described by (\ref{N}) and (\ref{Hlin}). However, unlike (\ref{bR}) and (\ref{H0lin}), the Hamiltonian $h_0$ is not assumed to be a quadratic function and is split into a quadratic part and a nonquadratic term $\wt{h}_0$. The nonquadratic part  $\wt{h}_0$ depends on $d$ system variables comprising a vector $Z X$ and is represented  in the Weyl quantization form with a Hermitian Fourier transform $\wt{H}_0: \mR^d \to \mC$:
\begin{align}
\label{hh}
    h_0  := b^{\rT} X + \frac{1}{2} X^{\rT} R X + \wt{h}_0,
    \qquad
    \wt{h}_0 := \int_{\mR^d} \wt{H}_0(v) \cW_{Z^{\rT}v} \rd v,
\end{align}
where $Z\in \mR^{d\x n}$ is a submatrix of a permutation matrix of order $n$ (so that $d\< n$ and $ZZ^{\rT} = I_d$).
For such a system, the drift vector $f$ and the dispersion matrix $g$ of the QSDE (\ref{dX}) in (\ref{fg}) take the form
\begin{equation}
\label{fglin}
    f = AX+2\Theta b + i[\wt{h}_0, X] = AX+2\Theta b + 2i\Theta Z^{\rT}\int_{\mR^d} \wt{H}_0(v) v\cW_{Z^{\rT}v}\rd v,
    \qquad
    g = B,
\end{equation}
where $A$ and $B$ are the matrices from (\ref{AB}) and Lemma~\ref{lem:WeylQSDE} is used. The nonquadratic part $\wt{h}_0$ of the Hamiltonian $h_0$ in (\ref{hh}) is the only source of the nonlinear dependence of $f$ on the system variables in (\ref{fglin}).

\begin{thm}
\label{th:Hlin}
Suppose the system-field coupling operators $h_1, \ldots, h_m$ are linear functions of the system variables described by (\ref{N}), (\ref{Hlin}), and the system Hamiltonian $h_0$ is decomposed according to (\ref{hh}). Then the IDE (\ref{Phidot}) for the QCF $\Phi$ in (\ref{char}) takes the form
\begin{align}
\nonumber
    \d_t \Phi(t,u)
    = &
    u^{\rT} A \d_u \Phi(t,u)
    +
    \Big(
        2iu^{\rT} \Theta b
        -\frac{1}{2}
        |B^{\rT}u|^2
    \Big)
    \Phi(t,u)
    \\
\label{PhidotHlin}
     &
     -
    2
    \int_{\mR^d}
    \sin(u^{\rT}\Theta Z^{\rT} v)
    \wt{H}_0(v)
    \Phi(t,u+Z^{\rT}v)\rd v    ,
\end{align}
where the matrices $A$ and $B$ are given by (\ref{AB}). The corresponding IDE for the QPDF $\mho$ in (\ref{Wig}) is
\begin{align}
\nonumber
    \d_t\mho(t,x) =&
    -\div(\mho(t,x)(Ax+2\Theta b)) + \frac{1}{2}\div^2(\mho(t,x) BB^{\rT})
\\
  \label{mhoIDE}
  &
  -
    2
    \int_{\mR^d}
    \Pi(x,v)\mho(t,x-\Theta Z^{\rT}v)
    \rd v
    ,
 \end{align}
 where the kernel function $\Pi: \mR^n\x \mR^d \to \mR$ is expressed in terms of the Fourier transform $\wt{H}_0$ from (\ref{hh}) as
 \begin{equation}
 \label{Pi}
    \Pi(x,v)
    :=
    \Im
    \big(
        \wt{H}_0(v)
        \re^{iv^{\rT} Z x }
    \big)
    =
    \Re \wt{H}_0(v) \sin(v^{\rT} Z x)
    +
    \Im \wt{H}_0(v) \cos(v^{\rT} Z x).
\end{equation}
\hfill $\square$
\end{thm}
\begin{proof}
We will use Theorem \ref{th:Phidot} and an intermediate step of its proof.
By substituting (\ref{fglin}) into (\ref{cA}) and using the fact that the system Hamiltonian enters the process $\cA_u$ in a linear fashion, it follows that
\begin{align}
\nonumber
    \cA_u
    & =
    iu^{\rT}
    \int_{-\frac{1}{2}}^{\frac{1}{2}}
    \Big(
    \cE_{r u}(AX + 2\Theta b + i[\wt{h}_0, X])
    +
    i
    \cE_{ru}(B)
    \Omega
    \int_r^{\frac{1}{2}}
    \cE_{su}(B)^{\rT}
    \rd s\,  u
    \Big)
    \rd r\\
\nonumber
    & =
    iu^{\rT}
    \int_{-\frac{1}{2}}^{\frac{1}{2}}
    \Big(
    A\cE_{r u}(X)
    + 2\Theta b
    +
    i
    \cE_{ru}([\wt{h}_0, X])
    +
    i \Big(\frac{1}{2}-r\Big)
    B
    \Omega
    B^{\rT}
    u
    \Big)
    \rd r\\
\nonumber
    & =
    iu^{\rT}
    \int_{-\frac{1}{2}}^{\frac{1}{2}}
    \Big(
    A(X+2r\Theta u)
    + 2\Theta b
    +
    i
    \cE_{ru}([\wt{h}_0, X])
    \Big)
    \rd r
    -
        \frac{1}{2}
        u^{\rT}
    B
    \Omega
    B^{\rT}
    u\\
\label{cAlin}
    & =
    iu^{\rT} (AX + 2\Theta b)
    -
    \frac{1}{2}|B^{\rT}u|^2
    -
    2
    \int_{\mR^d}
    \sin(u^{\rT}\Theta Z^{\rT}v) \wt{H}_0(v)
    \cW_{Z^{\rT}v} \rd v,
\end{align}
where $ u^{\rT}B\Omega B^{\rT}u = u^{\rT}B\cS(\Omega) B^{\rT}u = |B^{\rT}u|^2$ in view of (\ref{OmegaJ}).
By substituting (\ref{cAlin}) into  (\ref{cLcA}) and averaging, (\ref{PhidotEcL}) leads to
\begin{equation}
\label{PhidotHlin1}
    \d_t \Phi(t,u)
    =
    iu^{\rT}A
    \bE(\cW_{u/2}X \cW_{u/2})
    +
    \Big(
        2iu^{\rT} \Theta b
        -\frac{1}{2}
        |B^{\rT}u|^2
    \Big)
    \Phi(t,u)
      -
    2
    \int_{\mR^d}
    \sin(u^{\rT}\Theta Z^{\rT}v)
    \wt{H}_0(v)
    \Phi(t,u+Z^{\rT}v)\rd v.
\end{equation}
Now, (\ref{cE}), (\ref{cEuX}) and (\ref{XW}) imply that
$
    \cW_{u/2}X \cW_{u/2} = \cE_{u/2}(X)\cW_u = (X+ \Theta u) \cW_u = -i\d_u \cW_u
$,
which, in accordance with the first of the quasi-polynomial moment identities
(\ref{PhiX}), allows the expectation in (\ref{PhidotHlin1}) to be represented as $\bE(\cW_{u/2}X \cW_{u/2}) = -i\d_u \Phi(t,u)$, thus
establishing the IDE (\ref{PhidotHlin}). The IDE (\ref{mhoIDE}) can be obtained from  (\ref{PhidotHlin}) via the Fourier transform which relates the QPDF $\mho$ to the QCF $\Phi$. Therefore, the first line of (\ref{PhidotHlin}) leads to that of (\ref{mhoIDE}) in a standard fashion. In view of the inverse Fourier transform  in (\ref{Phimho}), the integral in (\ref{PhidotHlin}) can be represented in terms of $\mho$ as
\begin{align}
\nonumber
    \int_{\mR^d}
    \sin(u^{\rT}\Theta Z^{\rT}v)&
    \wt{H}_0(v)
    \Phi(t,u+Z^{\rT}v)
    \rd v\\
\nonumber
    & =
    \int_{\mR^d}
    \frac{\re^{iu^{\rT}\Theta Z^{\rT}v}-\re^{-iu^{\rT}\Theta Z^{\rT}v}}{2i}
    \wt{H}_0(v)
    \int_{\mR^n}
    \mho(t,y)
    \re^{i(u+Z^{\rT}v)^{\rT}y}
    \rd y
    \rd v\\
\nonumber
    & =
    \frac{1}{2i}
    \int_{\mR^n\x \mR^d}
    \big(\re^{iu^{\rT}(y+\Theta Z^{\rT}v)}-\re^{iu^{\rT}(y-\Theta Z^{\rT}v)}\big)
    \wt{H}_0(v)
    \mho(t,y)
    \re^{iv^{\rT}Z y}
    \rd y
    \rd v\\
\nonumber
    & =
    \frac{1}{2i}
    \int_{\mR^n\x \mR^d}
    \re^{iu^{\rT}x}
    \Big(
        \wt{H}_0(v) \re^{iv^{\rT}Z(x-\Theta Z^{\rT}v)}
        -
        \wt{H}_0(-v) \re^{-iv^{\rT}Z(x-\Theta Z^{\rT}v)}
    \Big)
    \mho(t,x-\Theta Z^{\rT}v)
    \rd x\rd v\\
\nonumber
    & =
    \frac{1}{2i}
    \int_{\mR^n\x \mR^d}
    \re^{iu^{\rT}x}
    \Big(
        \wt{H}_0(v) \re^{iv^{\rT}Zx}
        -
        \overline{\wt{H}_0(v) \re^{iv^{\rT}Zx}}
    \Big)
    \mho(t,x-\Theta Z^{\rT}v)
    \rd x\rd v\\
\label{intPi}
     & =
     \int_{\mR^n}
     \re^{iu^{\rT}x}
     \int_{\mR^d}
    \Pi(x,v)
    \mho(t,x-\Theta Z^{\rT}v)
    \rd v
    \rd x,
\end{align}
where $v^{\rT}Z\Theta Z^{\rT} v = 0$ since $\Theta$ is antisymmetric. Here, linear transformations of the integration variables $y\in \mR^n$, $v\in \mR^d$ are used together with the Hermitian property of the function $\wt{H}_0$ leading to the kernel function $\Pi$ in (\ref{Pi}) which satisfies
\begin{align}
\nonumber
    \Pi(y + \Theta Z^{\rT}v,\, v)
    & =
        \Im
    \big(
        \wt{H}_0(v)\re^{iv^{\rT}Z (y+ \Theta Z^{\rT}v)}
    \big)
    =
        \Im
    \big(
        \wt{H}_0(v)\re^{iv^{\rT}Z y}
    \big)\\
\label{Piyv}
    & =
    \Pi(y, v) = -\Pi(y,-v)
    =
-\Pi(y - \Theta Z^{\rT}v,\, -v).
\end{align}
The Fourier transform of the right-hand side of (\ref{intPi}) yields the integral operator term in (\ref{mhoIDE}), thus completing the proof of the theorem.
\end{proof}

The first line of the IDE (\ref{PhidotHlin}) is recognizable as a PDE which describes the QCF evolution for the system in the case $\wt{h}_0 = 0$:
\begin{equation}
\label{Phidotgauss}
    \d_t \Phi(t,u)
    =
    u^{\rT} A \d_u \Phi(t,u)
    +
    \Big(
        2iu^{\rT} \Theta b
        -\frac{1}{2}
        |B^{\rT}u|^2
    \Big)
    \Phi(t,u)
    =:
    \fA(\Phi(t,\cdot))(u),
\end{equation}
where $\fA$ is a linear first-order differential operator acting on the QCF $\Phi(t,\cdot)$.
 This corresponds to an open quantum harmonic oscillator of Section~\ref{sec:weyl}  governed by the linear QSDE (\ref{dXlin}), and the IDE (\ref{mhoIDE}) reduces to its first line
 \begin{equation}
 \label{mhodotgauss}
    \d_t\mho(t,x) = -\div(\mho(t,x)(Ax+2\Theta b)) + \frac{1}{2}\div^2(\mho(t,x) BB^{\rT}) =: \fF(\mho(t,\cdot))(x)
 \end{equation}
 which coincides with the FPKE for a classical Markov diffusion process with
the infinitesimal generator $\fF^{\dagger}$. The generator acts
on a smooth function $\varphi: \mR^n \to \mR$ (with the gradient $\varphi'$ and the Hessian matrix $\varphi''$) as  $
    \fF^{\dagger}(\varphi)(x):= (Ax+2\Theta b)^{\rT} \varphi'(x) + \frac{1}{2}\bra BB^{\rT}, \varphi''(x)\ket_{\rF}
 $.
  If such an oscillator is initialized at a Gaussian quantum state \cite{P_2010}, the linear dynamics preserve the Gaussian nature of the state in time \cite{JK_1998}. In the quantum optics literature,  this property is obtained for a one-mode oscillator from the corresponding master equation (see for example, \cite{GZ_2004}). In the mutidimensional case, the invariance of the class of Gaussian quantum states with respect to the linear dynamics follows from the PDEs (\ref{Phidotgauss}) and (\ref{mhodotgauss}) by the same reasoning as for classical linear stochastic systems.
 Indeed, the PDE (\ref{Phidotgauss}) can be  represented in terms of the logarithm of the QCF $\Phi$
 as a nonhomogeneous linear PDE
\begin{equation}
\label{logPhidotgauss}
    \d_t \ln \Phi(t,u)
    =
    u^{\rT} A \d_u \ln \Phi(t,u)
    +
        2iu^{\rT} \Theta b
        -\frac{1}{2}
        |B^{\rT}u|^2,
\end{equation}
which preserves the quadratic dependence of $\ln \Phi(t,u)$ on $u\in \mR^n$ over the course of time $t> 0$, provided $\ln \Phi(0,\cdot)$ is a quadratic function. In this case, the open quantum harmonic oscillator remains in the class of  Gaussian quantum states with the QCFs
\begin{equation}
\label{gausschar}
    \Phi_{\mu,\Sigma}(u) = \re^{i\mu^{\rT} u - \frac{1}{2}\|u\|_{\Sigma}^2},
\end{equation}
where $\mu$ is the mean vector of the system variables in (\ref{mu}), and $\Sigma$ is the real part of their quantum covariance matrix in (\ref{Sigma}). The QPDF $\mho$ in (\ref{Wig}), which corresponds to a Gaussian state with the QCF $\Phi_{\mu,\Sigma}$ in (\ref{gausschar}), is given by
\begin{equation}
\label{gaussWig}
    \mho_{\mu,\Sigma}(x) := \frac{(2\pi)^{-n/2}}{\sqrt{\det \Sigma}}\re^{-\frac{1}{2}\|x-\mu\|_{\Sigma^{-1}}^2},
    \qquad
    x\in \mR^n,
\end{equation}
provided $\Sigma\succ 0$. From
(\ref{logPhidotgauss}), it follows  that the parameters $\mu$ and $\Sigma$ of the Gaussian quantum state satisfy the ODEs
\begin{equation}
\label{muSigmadot}
    \dot{\mu} = A \mu + 2\Theta b,
    \qquad
    \dot{\Sigma} = A\Sigma + \Sigma A^{\rT} + BB^{\rT}.
\end{equation}
Here, the second equation is a Lyapunov ODE whose solution satisfies $\Sigma(t)\succ 0$ at all times $t>0$ for any physically meaningful initial data\footnote{with $\Sigma(0)\succcurlyeq -i\Theta$ due to the positive semi-definiteness of the quantum covariance matrix in (\ref{covX}) as a generalized form of the Heisenberg uncertainty principle \cite{H_2001}} if the Krylov subspaces, generated by the matrix $A$ from the columns of $B$, cover the space $\mR^n$. The latter condition is equivalent to the Kalman controllability matrix $\begin{bmatrix}  B & AB & \ldots & A^{n-1}B\end{bmatrix}$ being of full row rank, which is closely related to the H\"{o}rmander condition \cite{H_1967,S_2008} and guarantees that the PDE (\ref{mhodotgauss}) has a smooth fundamental solution. Irrespective of whether the matrices $A$ and $B$ satisfy the controllability condition, the solutions
\begin{equation}
\label{muSigma}
    \mu(t)= 2\int_0^t \re^{sA}\rd s\,  \Theta b,
    \qquad
    \Sigma(t) = \int_0^t \re^{sA}BB^{\rT}\re^{sA^{\rT}}\rd s
\end{equation}
of the ODEs (\ref{muSigmadot}) (this time with zero initial conditions $\mu(0)=0$ and $\Sigma(0) = 0$) parameterize the solutions of the PDE (\ref{Phidotgauss}) for general (not necessarily  Gaussian) initial QCFs $\Phi(0,\cdot)$. More precisely, application of the method of characteristics \cite{E_1998,V_1971} allows $\Phi$ to be expressed in terms of (\ref{gausschar}) as
\begin{equation}
\label{Phisol}
    \Phi(t,u) = \Phi\big(0,\re^{tA^{\rT}}u\big) \Phi_{\mu(t), \Sigma(t)}(u) = \re^{t\fA} (\Phi(0,\cdot))(u),
\end{equation}
where the linear operator $\re^{t\fA}$ 
describes the flow associated with the PDE (\ref{Phidotgauss}). Since any QCF takes values in the closed unit disc of the complex plane,  and the Gaussian QCF satisfies  $|\Phi_{\mu,\Sigma}(u)|^2 = \re^{-\|u\|_{\Sigma}^2}$, the representation (\ref{Phisol}) implies the finiteness of the following weighted $L_2$-norm
\begin{equation}
\label{Phinorm}
    \sn \Phi(t,\cdot)\sn_P
    := \sqrt{\int_{\mR^n}\re^{u^{\rT}Pu}|\Phi(t,u)|^2\rd u}
\end{equation}
for any matrix $P\in \mS_n$ satisfying $P\prec  \Sigma(t)$. The corresponding inner product $\dbra \varphi, \psi\dket_P := \int_{\mR^n} \re^{u^{\rT}Pu} \overline{\varphi(u)} \psi(u)\rd u$ is real-valued for Hermitian functions $\varphi$ and $\psi$.
Since the finiteness of the integral in (\ref{Phinorm}) for any $P\prec  0$ is trivially ensured by the inequality $|\Phi|\< 1$ (see also
\cite[Proposition 6]{CH_1971} on square integrability of QCFs), we will assume for what follows that $P\succcurlyeq 0$. For example, by letting $P:= \Sigma(t)$ in (\ref{Phinorm}) and combining (\ref{gausschar}) with (\ref{Phisol}), it follows that
$$
    \sn\Phi(t,\cdot)\sn_{\Sigma(t)}^2
    =
    \int_{\mR^n}
    \big|\Phi(0,\re^{tA^{\rT}}u)\big|^2
    \rd u
    =
    \|\Phi(0,\cdot)\|_2^2\,
        \re^{-t\Tr A}.
$$
If the pair $(A,B)$ is controllable, then for any $t>0$, the controllability Gramian $\Sigma(t)$ in (\ref{muSigma}) is positive definite and hence, there exist $P\in \mS_n$ satisfying $0\prec P \prec \Sigma(t)$. From the finiteness of the norm $\sn \Phi(t,\cdot)\sn_P$ for such matrices $P$ in (\ref{Phinorm}) and the Plancherel identity, it follows that the QPDF $\mho(t,\cdot)$ is infinitely differentiable, and its partial derivatives are all square integrable over $\mR^n$, with
\begin{align}
\nonumber
    \sum_{\alpha \in \mZ_+^n}
    \frac{\eps^{|\alpha|}}{\alpha !}
    \|\d_x^{\alpha}\mho(t,\cdot)\|_2^2
    & =
    (2\pi)^{-n}
    \int_{\mR^n}
    |\Phi(t,u)|^2
    \sum_{\alpha \in \mZ_+^n}
    \frac{\eps^{|\alpha|}}{\alpha !}
    u^{2\alpha}
    \rd u\\
\label{Sob}
    & =
    (2\pi)^{-n}
    \int_{\mR^n}
    \re^{\eps |u|^2}
    |\Phi(t,u)|^2\rd u
    \<
    (2\pi)^{-n}
    \sn \Phi(t,\cdot)\sn_P^2<+\infty
\end{align}
for any real $\eps$ not exceeding  the smallest eigenvalue of $P$.
The quantity on the left-hand side of (\ref{Sob}) coincides with the square of a generalized Sobolev  norm $\|\re^{-\frac{\eps}{2}\Delta}\mho(t,\cdot)\|_2$ of the QPDF, which is organised similarly to the norm in the Bessel potential space \cite[pp. 170--171]{S_2008}. Here, $\Delta$ is the Laplacian, and the operator $\re^{-\frac{\eps}{2}\Delta}$ (with $\eps>0$) describes the time-reversed flow
associated with the heat equation $\d_t \psi = \frac{1}{2}\Delta \psi$.
The property (\ref{Sob}) of the solutions of the PDEs (\ref{Phidotgauss}) and (\ref{mhodotgauss}) for the open quantum harmonic oscillator are qualitatively  related to the smoothing effect of the heat semigroup \cite{E_1998}.\footnote{a function $\varphi: \mR^n \to \mC$ satisfies $\sn \varphi\sn_{\eps I_n}^2<+\infty$  for some $\eps>0$ if and only if its Fourier transform is a solution of the heat equation at time $\eps$ with a square integrable initial condition}
However, in contrast to the linear case, the presence of a nonquadratic term in the system Hamiltonian (\ref{hh}) makes  the quantum state dynamics non-Gaussian, and the smoothness of fundamental solutions of the IDE (\ref{mhoIDE}) requires a separate investigation. Although this is beyond the scope of the present paper, we note that such analysis can employ a Duhamel type formula \cite{E_1998,S_2008}
\begin{equation}
\label{duhamel}
    \Phi(t,\cdot) = \re^{t(\fA+\fB)}(\Phi(0,\cdot)) = \re^{t\fA}(\Phi(0,\cdot)) + \int_0^t \re^{(t-s)\fA}(\fB(\Phi(s,\cdot))) \rd s,
\end{equation}
where $\fB$ denotes the integral operator on the right-hand side of (\ref{PhidotHlin}) which acts over the spatial variables of the QCF as
\begin{equation}
\label{fB}
    \fB(\varphi)(u)
     :=
     -
    2
    \int_{\mR^d}
    \sin(u^{\rT}\Theta Z^{\rT} v)
    \wt{H}_0(v)
    \varphi(u+Z^{\rT}v)\rd v,
\end{equation}
with the image $\fB(\Phi(t,\cdot))(u)$ also being a Hermitian function of $u \in \mR^n$ for any given time $t\>0$. We will only outline
a possible avenue for obtaining upper bounds for the norms  (\ref{Phinorm}) in the non-Gaussian case, based on  the dissipation relation
 \begin{equation}
 \label{diss}
    \d_t (\sn \Phi\sn_P^2)
        +
    \Bra
        AP+PA^{\rT}+BB^{\rT},
        \d_P(\sn \Phi\sn_P^2)
    \Ket_{\rF}
    =
    -\sn \Phi\sn_P^2 \Tr A
    +
    2\dbra \Phi, \fB(\Phi)\dket_P,
 \end{equation}
where $\d_P(\cdot)$ is the Frechet derivative on the Hilbert space $\mS_n$ endowed with the Frobenius inner product of matrices $\bra\cdot, \cdot\ket_{\rF}$. The relation (\ref{diss}) is established by computing the partial time derivative of the squared norm from (\ref{Phinorm})  as
\begin{align}
\nonumber
    \d_t (\sn \Phi\sn_P^2)
    = &
    2\dbra \Phi, \d_t \Phi\dket_P\\
\nonumber
     = &
     2
     \int_{\mR^n}
     \re^{\|u\|_P^2}\,
        \overline{\Phi(t,u)}
        \Big(u^{\rT} A \d_u \Phi(t,u)
        +
        \big( 2iu^{\rT} \Theta b
        -
        \frac{1}{2}
        |B^{\rT}u|^2
        \big)
        \Phi(t,u)
        +
        \fB(\Phi(t,\cdot))(u)\!
        \Big)
        \rd u\\
\nonumber
    = &
    -\int_{\mR^n}
    \re^{\|u\|_P^2}
    \big(\Tr A + u^{\rT}(AP + PA^{\rT} + BB^{\rT})u\big)
    |\Phi(t,u)|^2\rd u
    +
    2\dbra \Phi, \fB(\Phi)\dket_P\\
\label{diss_der}
    = &
    -\sn \Phi\sn_P^2 \Tr A
    -
    \Bra
        AP+PA^{\rT}+BB^{\rT},\,
        \d_P(\sn \Phi\sn_P^2)
    \Ket_{\rF} +
    2\dbra \Phi, \fB(\Phi)\dket_P,
\end{align}
provided the QCF $\Phi(t,\cdot)$  decays fast enough at infinity in $\mR^n$ in the sense that $|\Phi(t,u)|^2 = o(\re^{-\|u\|_P^2}|u|^{-n})$ as $u\to \infty$. This decay rate condition is combined in (\ref{diss_der}) with the divergence theorem and
the identities
\begin{align*}
    2\re^{\|u\|_P^2} u^{\rT}A
    \Re (\overline{\varphi} \d_u \varphi)
    & =
    \div\big(\re^{\|u\|_P^2}A^{\rT} u |\varphi|^2\big) - |\varphi|^2 \div\big(\re^{\|u\|_P^2}A^{\rT} u\big),\\
    \div\big(\re^{\|u\|_P^2}A^{\rT} u \big)
    & =
    \re^{\|u\|_P^2}(\Tr A + u^{\rT}(AP + PA^{\rT})u),\\
    \d_P\,  \re^{\|u\|_P^2} & = \re^{\|u\|_P^2} uu^{\rT}.
\end{align*}
The first-order differential operator on the left-hand side of (\ref{diss}) is the full time derivative of $\sn \Phi(t,\cdot)\sn_P^2$ (as a function of $(t,P)\in \mR_+\x \mS_n$) along the characteristics $\dot{P} = AP+PA^{\rT} + BB^{\rT}$.
This allows a multivariate partial differential version  of Gronwall's lemma to be applied to $\sn \Phi\sn_P^2$, provided the matrix $P$ is not too ``large'' in comparison with the controllability Gramian $\Sigma$ in (\ref{muSigma}).
To this end, in view of (\ref{fB}), the rightmost term in (\ref{diss}) admits the following upper bound
\begin{align}
\nonumber
    \dbra \varphi, \fB(\varphi)\dket_P
    & =
     -
    2
    \int_{\mR^n}
    \re^{\|u\|_P^2}
    \overline{\varphi(u)}
    \int_{\mR^d}
    \sin(u^{\rT}\Theta Z^{\rT} v)
    \wt{H}_0(v)
    \varphi(u+Z^{\rT}v)\rd v
    \rd u\\
\nonumber
    & \<
    2
    \int_{\mR^n\x \mR^d}
    \re^{-\|u\|_{S-P}^2-u^{\rT} S Z^{\rT} v - \frac{1}{2}\|v\|_{ZSZ^{\rT}}^2}
    |\wt{H}_0(v)|
    \varphi_S(u)
    \varphi_S(u+Z^{\rT}v)
    \rd u
    \rd v\\
\label{EEE}
    & \<
    2
        \sn\varphi\sn_S^2
    \int_{\mR^d}
    \re^{\frac{1}{4}v^{\rT}Z(P + P(S-P)^{-1}P - S)Z^{\rT}v}
    |\wt{H}_0(v)|
    \rd v,
\end{align}
where $S \in \mS_n$ is an arbitrary matrix satisfying $S\succ P$. Here, the  Cauchy-Bunyakovsky-Schwarz inequality is applied to the function $\varphi_S(u):=    \re^{\frac{1}{2}\|u\|_S^2}
    |\varphi(u)|$ and its translation $\varphi_S(u+Z^{\rT}v)$, so that $\int_{\mR^n}\varphi_S(u)\varphi_S(u+Z^{\rT}v)\rd u\< \|\varphi_S\|_2^2 = \sn \varphi\sn_S^2$ for all $v\in \mR^d$. 
    Also,  use is made of the relations
    $$\|u\|_{S-P}^2+u^{\rT} S Z^{\rT} v = \|u+\frac{1}{2} (S-P)^{-1}S Z^{\rT}v\|_{S-P}^2 - \frac{1}{4}\|SZ^{\rT}v\|_{(S-P)^{-1}}^2\> - \frac{1}{4}\|SZ^{\rT}v\|_{(S-P)^{-1}}^2
    $$
    (which are obtained by completing the square) and the matrix identity $S(S-P)^{-1}S-2S = 
    P + P(S-P)^{-1}P - S$. While $\sn\varphi\sn_S$ is an increasing  function  of the auxiliary matrix  $S$ (in the sense of the matrix ordering), the last integral in (\ref{EEE}) decreases with respect to  $S$. This integral is amenable to an explicit calculation when $|\wt{H}_0(v)|$ is a quadratic-exponential function of $v$, which is the case, for example,  if the non-quadratic part $\wt{h}_0$ of the Hamiltonian in (\ref{hh}) is the Weyl quantization of a Gaussian-shaped function on $\mR^d$. Since the  matrix $S\succ P$ is arbitrary, (\ref{EEE})  leads to
\begin{equation}
\label{EEE1}
    \dbra \varphi, \fB(\varphi)\dket_P
     \<
    2
    \inf_{S\in \mS_n:\ S \succ P}
    \Big(
        \sn\varphi\sn_S^2
    \int_{\mR^d}
    \re^{\frac{1}{4}v^{\rT}Z(P + P(S-P)^{-1}P - S)Z^{\rT}v}
    |\wt{H}_0(v)|
    \rd v
    \Big)
    \<
    2
        \sn\varphi\sn_{2P}^2
    \int_{\mR^d}
    |\wt{H}_0(v)|
    \rd v.
\end{equation}
However, the difficulty of using the upper bounds (\ref{EEE}) and (\ref{EEE1}) lies in the fact that they involve the norms $\sn \cdot \sn_S$ with different matrices $S\succ P$, which couples  the dissipation inequalities for  $\sn \Phi \sn_P$ resulting from a combination of the bounds with (\ref{diss}).

In conclusion of this section, we note that the IDE (\ref{mhoIDE}) suggests an analogy with the PDF dynamics of classical jump-diffusion processes, especially considering the fact that, in view of the antisymmetry (\ref{Piyv}), 
the integral operator term in (\ref{mhoIDE}) satisfies the identity
$$
    \int_{\mR^n}
    \int_{\mR^d}
    \Pi(x,v) \mho(t,x-\Theta Z^{\rT}v)
    \rd v
    \rd x
    =
    \int_{\mR^n}
        \mho(t,y)
    \int_{\mR^d}
    \Pi(y + \Theta Z^{\rT}v, v)
    \rd v
    \rd y
=    0,
$$
which is closely related to the preservation of the normalization (\ref{one}), similar to the property of classical PDFs. However,
even in the nondegenerate case when $d=n$ and $\det \Theta \ne 0$,  the integral operator $-2\int_{\mR^d} \Pi(x,v)\mho(t,x-\Theta Z^{\rT}v)\rd v$ in (\ref{mhoIDE}) is not necessarily representable  in the form $\int_{\mR^n} p(x\mid y)\lambda(y)\mho(t,y)\rd y - \lambda(x)\mho(t,x)$  which corresponds to a jump-diffusion process whose ``jump'' part is specified by a state-dependent rate $\lambda$ and an absolutely continuous Markov transition kernel with a conditional PDF $p$.   This discrepancy and its interplay with the FPKE part of the IDE (\ref{mhoIDE}) can lead to negative values of the QPDF $\mho$.

\section{Approximate computation of invariant states via operator splitting}\label{sec:inv}

For the class of quantum systems with linear system-field coupling, described in the previous section,   we will now consider the problem of finding invariant states in the form of steady-state solutions of the IDEs (\ref{PhidotHlin}) and (\ref{mhoIDE}) which can be written as
$$
    \d_t\Phi = (\fA+\fB)(\Phi),
    \qquad
    \d_t\mho = (\fF+\fG)(\mho),
$$
where the first of the equations has already been used in (\ref{duhamel}). The integral operators $\fB$ and $\fG$  on the right-hand sides of these IDEs act over the spatial variables of the QCF and QPDF according to (\ref{fB}) and
\begin{equation}
  \label{fG}
    \fG(\mho(t,\cdot))(x)
     :=  -
    2
    \int_{\mR^d}
    \Pi(x,v)\mho(t,x-\Theta Z^{\rT}v)
    \rd v.
 \end{equation}
 By using the general idea of the operator-splitting methods \cite{Mar_1988,S_1968}, the operators $\fB$ and $\fG$
can be regarded as perturbations to the differential operators $\fA$ and $\fF$ of the exactly solvable PDEs (\ref{Phidotgauss}) and (\ref{mhodotgauss}) which have Gaussian steady-state solutions. More precisely, if the matrix $A$ in (\ref{AB}) is Hurwitz, then the open quantum harmonic oscillator (\ref{dXlin}), which represents the linear part  of the system, has a Gaussian invariant state whose QCF $\Phi_0$ and QPDF $\mho_0$ are given by
\begin{equation}
\label{Phimho0}
    \Phi_0 := \Phi_{\mu_0,\Sigma_0},
    \qquad
    \mho_0 := \mho_{\mu_0,\Sigma_0}.
\end{equation}
Here, in view of (\ref{gausschar})--(\ref{muSigma}), the mean vector $\mu_0$ and the real part $\Sigma_0$ of the quantum covariance matrix are computed as
\begin{equation}
\label{muSigma0}
    \mu_0 := -2A^{-1} \Theta b,
    \qquad
    \Sigma_0 := \int_0^{+\infty}\re^{tA} BB^{\rT}\re^{tA^{\rT}}\rd t,
\end{equation}
with $\Sigma_0$ being the unique solution of the algebraic Lyapunov equation $A\Sigma_0 + \Sigma_0A^{\rT} + BB^{\rT} = 0$.
The functions $\Phi_0$ and $\mho_0$ in (\ref{Phimho0}) provide initial approximations to the invariant state of the nonlinear quantum system. The invariant QCF $\Phi_*$ and the invariant QPDF $\mho_*$ of the system can then be sought as the formal series
\begin{equation}
\label{Phimhoinf}
    \Phi_* = \sum_{k=0}^{+\infty} \Phi_k, 
    \qquad
    \mho_* = \sum_{k=0}^{+\infty} \mho_k. 
\end{equation}
The terms $\Phi_k: \mR^n \to \mC$ and their Fourier transforms  $\mho_k:\mR^n\to \mR$ are computed through the recurrence relations
\begin{equation}
\label{Phimhonext}
    \fA(\Phi_k) + \fB(\Phi_{k-1}) = 0,
    \qquad
    \fF(\mho_k) + \fG(\mho_{k-1}) = 0,
\end{equation}
which are organized as nonhomogeneous linear PDEs with respect to $\Phi_k$ and $\mho_k$ subject to the
normalization constraints
\begin{equation}
\label{Phimhonorm}
  \Phi_k(0)=0,
  \qquad
  \int_{\mR^n}
  \mho_k(x)\rd x
  =0
\end{equation}
for all $k=1,2,3,\ldots$, with the Gaussian initial conditions $\Phi_0$ and $\mho_0$ given by (\ref{Phimho0}) and (\ref{muSigma0}). Although the differential operators $\fA$ and $\fF$ themselves are not invertible, the solutions of the equations (\ref{Phimhonext}) and (\ref{Phimhonorm}) are formally representable as
$
    \Phi_k =
(-\fA^{-1} \fB)^k(\Phi_0)$
and
$
    \mho_k = (-\fF^{-1} \fG)^k(\mho_0)
$,
and hence, the convergence of the series in (\ref{Phimhoinf}) depends on the decay of iterates of the operators $\fA^{-1} \fB$ and $\fF^{-1} \fG$ on $\Phi_0$ and $\mho_0$ in an appropriate sense.

As an illustrative example concerning the first perturbation terms $\Phi_1$ and $\mho_1$, suppose (throughout the rest of this section) that the system variables consist of $d:=\frac{n}{2}$ Cartesian position variables $q_1, \ldots, q_d$ and the conjugate momentum operators $p_1, \ldots, p_d$:
\begin{equation}
\label{Xqp}
    X:= \begin{bmatrix}q\\ p\end{bmatrix},
    \qquad
    q:= \begin{bmatrix}q_1\\ \vdots\\ q_d\end{bmatrix} = ZX,
    \qquad
    Z = \begin{bmatrix}I_d & 0\end{bmatrix},
    \qquad
    p:= \begin{bmatrix}p_1\\ \vdots\\ p_d\end{bmatrix} = -i\d_q,
\end{equation}
with the CCR matrix $\Theta$ given by (\ref{ThetaJ}). Furthermore, let the system Hamiltonian $h_0$ be described by (\ref{hh}) with $b=0$ as
\begin{equation}
\label{Gamma}
    h_0 := \frac{1}{2}\big(q^{\rT}\Gamma q + p^{\rT}p\big) +  \phi(q) = \frac{1}{2}X^{\rT} R X + \wt{h}_0,
    \qquad
    R:= \begin{bmatrix} \Gamma & 0\\ 0 & I_d\end{bmatrix},
\end{equation}
where the quadratic part of the total energy is specified by the stiffness matrix $\Gamma \in \mS_n$ and the identity mass matrix $I_d$, while the nonquadratic part of the potential energy is described by a function $\phi:\mR^d\to \mR$ with the Fourier transform $\wt{H}_0$:
\begin{equation}
\label{phi}
    \wt{h}_0 := \phi(q) = \int_{\mR^d} \wt{H}_0(v)\re^{iv^{\rT}q}\rd v.
\end{equation}
In accordance with (\ref{hh}), the matrix $Z$ in (\ref{Xqp}) consists of the first $d$ rows of $I_n$,  and use is made of the mutual commutativity of the position variables in the vector $q$, whereby the corresponding Weyl operator $\cW_{Z^{\rT} v}$ reduces to the usual exponential function  $\re^{iv^{\rT} ZX} = \re^{iv^{\rT}q}$ for any $v \in \mR^d$. Now, consider a negative Gaussian-shaped potential (see, for example, \cite{FM_2004} and references therein):
\begin{equation}
\label{phiMorse}
    \phi(q) := -C\re^{-\frac{1}{2}\|q-\gamma\|_{\Lambda}^2},
\end{equation}
where $C>0$, $\gamma \in \mR^d$, and $\Lambda\in   \mS_d$ is a positive definite matrix. The parameter $\gamma$ of the potential $\phi$
specifies the location of an attracting centre in the position space $\mR^d$ with the stiffness matrix $\phi''(\gamma) = C\Lambda$, with the exponentially fast decay of the attraction at infinity resembling the Morse potential \cite{M_1929} (see Fig.~\ref{fig:morse} for a two-dimensional example of the potential energy in (\ref{Gamma}) with $\Gamma={\small\begin{bmatrix}0.2 & -0.1\\ -0.1 & 0.4\end{bmatrix}}$, $C=1.5$, $\Lambda={\small\begin{bmatrix}6 & 1\\ 1 & 4\end{bmatrix}}$ and $\gamma = {\small\begin{bmatrix}2 \\ 3\end{bmatrix}}$).
\begin{figure}[thpb]
      \centering
      \includegraphics[width=110mm]{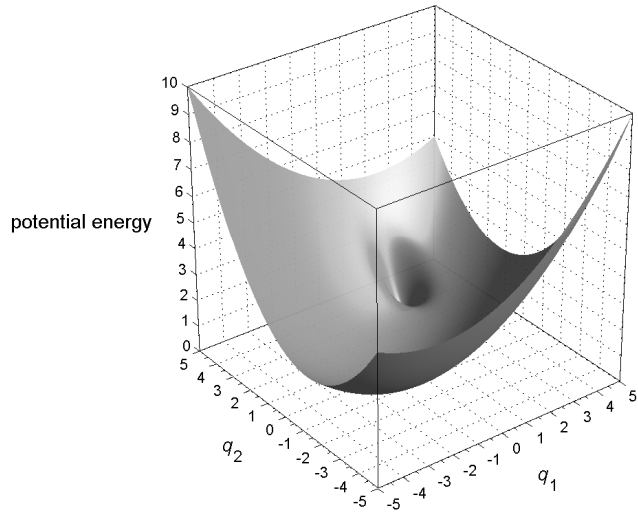}
      \caption{A two-dimensional illustration of the quadratic and negative Gaussian   potential energy  $\frac{1}{2}\|q\|_{\Gamma}^2 + \phi(q)$,   with $\Gamma\succ 0$ and $\phi$ given by (\ref{phiMorse}), as a function of the position vector $q\in \mR^2$. The second local minimum is contributed by $\phi$.}
      \label{fig:morse}
   \end{figure}
The corresponding function $\wt{H}_0$ in (\ref{phi}) is the Fourier transform of (\ref{phiMorse}):
$$
    \wt{H}_0(v)
    = (2\pi)^{-d}\int_{\mR^d}\phi(q)\re^{-iv^{\rT}q}\rd q
    = -C \frac{(2\pi)^{-d/2}}{\sqrt{\det \Lambda}}
    \re^{-iv^{\rT}\gamma-\frac{1}{2}\|v\|_{\Lambda^{-1}}^2},
$$
whose substitution into (\ref{Pi}) leads to the following kernel function of the integral operator $\fG$ in  (\ref{fG}):
\begin{equation}
\label{Piex}
    \Pi(x,v)
    =
    C \frac{(2\pi)^{-d/2}}{\sqrt{\det \Lambda}}
    \re^{-\frac{1}{2}\|v\|_{\Lambda^{-1}}^2}
    \sin(v^{\rT}(\gamma-q)).
\end{equation}
Assuming that the pair $(A,B)$ is controllable in addition to the matrix $A$ being Hurwitz, the parameters (\ref{muSigma0}) of the invariant Gaussian state of the linear part of the system satisfy $\mu_0 = 0$ and $\Sigma_0 \succ 0$. 
The image of the Gaussian QPDF $\mho_0 = \mho_{0,\Sigma_0}$ from (\ref{Phimho0}) under the integral operator $\fG$ in (\ref{fG}), associated with  (\ref{Piex}), can be  computed as
\begin{align}
\nonumber
    \fG(\mho_0)(x)
    & =
    2C \frac{(2\pi)^{-(n+d)/2}}{\sqrt{\det \Lambda\det\Sigma_0}}
    \int_{\mR^d}
    \re^{-\frac{1}{2}(\|v\|_{\Lambda^{-1}}^2 + \|x - \Theta Z^{\rT} v\|_{\Sigma_0^{-1}}^2)}
    \sin(v^{\rT}(q-\gamma))
    \rd v\\
\nonumber
& =
    2C \frac{(2\pi)^{-(n+d)/2}}{\sqrt{\det \Lambda\det\Sigma_0}}
    \re^{-\frac{1}{2}
    \|x\|_{(\Sigma_0 - \Theta Z^{\rT} \Lambda Z \Theta )^{-1}}^2}
    \Im
    \int_{\mR^d}
    \re^{iv^{\rT}(q-\gamma)-\frac{1}{2}\|v+ SZ\Theta \Sigma_0^{-1}x\|_{S^{-1}}^2
}
    \rd v\\
\nonumber
& =
    \frac{2C (2\pi)^{-d}}{\sqrt{\det\Sigma_0\det(I_d - \Lambda Z\Theta \Sigma_0^{-1} \Theta Z^{\rT})}}
    \re^{-\frac{1}{2}
    \|x\|_{(\Sigma_0 - \Theta Z^{\rT} \Lambda Z \Theta )^{-1}}^2}
    \Im
    \re^{i(\gamma-q)^{\rT}SZ\Theta \Sigma_0^{-1}x -\frac{1}{2}\|\gamma-q\|_S^2}
\\
\label{fGmho0}
& =
    E
    \re^{\sigma^{\rT}x -\frac{1}{2}\|x\|_{\alpha}^2}
    \sin
    \Big(\tau^{\rT} x - \frac{1}{2}x^{\rT} \beta x\Big)
\end{align}
for any $x\in \mR^n$, with $q=Zx$ in view of (\ref{Xqp}).
Here, $E>0$ is a constant factor and $S\in \mS_n$ is a positive definite matrix given by
$$
    E:=
    \frac{2C (2\pi)^{-d}\, \re^{-\frac{1}{2}\|\gamma\|_S^2}}{\sqrt{\det\Sigma_0\det(I_d - \Lambda Z\Theta \Sigma_0^{-1} \Theta Z^{\rT})}},
    \qquad
S:= (\Lambda^{-1} - Z\Theta \Sigma_0^{-1} \Theta Z^{\rT})^{-1},
$$
and use is made of
the relation
$
    \Sigma_0^{-1} + \Sigma_0^{-1} \Theta Z^{\rT}S Z\Theta \Sigma_0^{-1}
    =
    (\Sigma_0 - \Theta Z^{\rT} \Lambda Z \Theta )^{-1}
$
which follows
from the Sherman-Morrison-Woodbury
matrix identity  \cite{HJ_2007}. Also, $\alpha, \beta \in \mS_n$, $\sigma, \tau \in \mR^n$ are auxiliary matrices and vectors given by
\begin{align*}
    \alpha
     & :=
    (\Sigma_0 - \Theta Z^{\rT} \Lambda Z \Theta )^{-1}
    +
    Z^{\rT}SZ,
    \qquad
    \beta
     :=
    Z^{\rT}
    S
    Z\Theta \Sigma_0^{-1}
    -
    \Sigma_0^{-1}\Theta Z^{\rT}
    S
    Z,\\
    \sigma
     & :=
    Z^{\rT}S\gamma,
    \qquad\qquad\qquad\qquad\qquad\quad
    \tau
     :=
    -\Sigma_0^{-1}\Theta Z^{\rT}S \gamma.
\end{align*}
Since $\alpha \succ 0$,
the oscillatory quadratic-exponential function $\fG(\mho_0)(x)$ in (\ref{fGmho0}) has a Gaussian decay rate
as $x\to \infty$.  In view of (\ref{mhodotgauss}), (\ref{Phimhonext}) and (\ref{Phimhonorm}), the first correction $\mho_1$ of $\mho_0$ towards the invariant QPDF $\mho_*$ in (\ref{Phimhoinf}) is found by solving the problem
\begin{equation}
\label{PDE}
    \fF(\mho_1)(x)
    :=
    \div\Big(\frac{1}{2}BB^{\rT}\mho_1'(x) - \mho_1(x)Ax\Big)
    =
    -\fG(\mho_0)(x),
    \qquad
    \int_{\mR^n}
    \mho_1(x)\rd x = 0.
\end{equation}
The Green's function $\kappa: \mR^n\x \mR^n\to \mR$ for this problem 
is expressed in terms of the transitional PDF of the classical Markov diffusion process with the generator $\fF^{\dagger}$ as
\begin{equation}
\label{kappa}
    \kappa(x,y)
    :=
        \int_0^{+\infty}
    \big(
        \mho_{\re^{tA} y,\, \Sigma(t)}(x)-\mho_0(x)
    \big)
    \rd t,
\end{equation}
where use is made of the Gaussian PDF (\ref{gaussWig}) and the finite-horizon controllability Gramian $\Sigma(t)$ from (\ref{muSigma}). The convergence of the improper integral in (\ref{kappa}) at $+\infty$ is secured by the exponentially fast convergence of $\re^{tA}$ to $0$ and $\Sigma(t)$ to $\Sigma_0$ as $t\to +\infty$ due to the matrix $A$ being Hurwitz. The convergence of this integral at $0$ can be ensured by an additional assumption of ellipticity  $BB^{\rT}\succ 0$, which is equivalent to the matrix $B$ being of full row rank and is stronger than the controllability of $(A,B)$.
Note that $\int_{\mR^n}\kappa(x,y)\rd x = 0$ for all $y\in \mR^n$.
The solution of the problem (\ref{PDE}) takes the form
\begin{equation}
\label{Fres}
    \mho_1(x)
    =
    \int_{\mR^n}
    \kappa(x,y)
    \fG(\mho_0)(y)
    \rd y.
\end{equation}
Although, in view of (\ref{fGmho0}), the right-hand side of (\ref{Fres}) resembles the structure of Fresnel integrals \cite{V_2002} with a Gaussian damping, its calculation in closed form is complicated by the presence of the integration over time in (\ref{kappa}).

The above phase-space approach to the approximation of invariant quantum states as steady-state solutions of the IDEs for the QCF and QPDF dynamics
can be extended to the fundamental solutions of the IDEs (which are quantum counterparts of the classical Markov transition kernels and  specify the relaxation dynamics of the system towards the equilibrium).

\section{$\chi^2$-divergence from Gaussian states}\label{sec:chi2}

For the quantum systems with linear system-field coupling from Section~\ref{sec:Hlin},
we will now consider the deviation of the actual quantum state from Gaussian states which the system would have if its dynamics (\ref{dX}) were linear and the initial state were Gaussian. This deviation can be quantified by the $\chi^2$-divergence
\begin{equation}
\label{chi2}
    D_{\mu,\Sigma}(\mho)
    :=
    \int_{\mR^n}
    \frac{(\mho-\mho_{\mu,\Sigma})^2}{\mho_{\mu,\Sigma}}\rd x
    =
    \int_{\mR^n}
    \frac{\mho^2}{\mho_{\mu,\Sigma}}\rd x
    -
    1
    =
    \re^{\bR(\mho\| \mho_{\mu,\Sigma})}-1
\end{equation}
which, unlike the standard Kullback-Leibler relative entropy \cite{CT_1991,G_2008}, is well-defined despite possible negative values of the QPDF $\mho$.
We have omitted the arguments of the functions for brevity and used the normalization (\ref{one}) which holds for an arbitrary QPDF $\mho$, including its Gaussian case $\mho_{\mu,\Sigma}$ in (\ref{gaussWig}).
The $\chi^2$-divergence in (\ref{chi2}) is expressed in terms of the second-order Renyi relative entropy \cite{R_1961} of a PDF $p: \mR^n\to \mR_+$ with respect to a reference PDF $r: \mR^n\to \mR_+$  defined by
$$
    \bR(p\| r):= \ln \int_{\mR^n}\frac{p^2}{r}\rd x = \ln \int_{\mR^n}\left(\frac{p}{r}\right)^2 r\rd x,
$$
provided
$p$ is absolutely continuous with respect to $r$ (in the sense that $p=0$ whenever $r=0$). Although the QPDF  $\mho$ can take negative values, the Renyi relative entropy retains its usefulness as a measure of deviation in (\ref{chi2}) due to (\ref{one}) and the fact that $\mho_{\mu,\Sigma}$ is a legitimate PDF. This follows from the Cauchy-Bunyakovsky-Schwarz inequality
$$
    1
    =
    \Big(
    \int_{\mR^n}\frac{\mho}{\sqrt{\mho_{\mu,\Sigma}}}\sqrt{\mho_{\mu,\Sigma}}\ \rd x
    \Big)^2
    \<
    \int_{\mR^n}\frac{\mho^2}{\mho_{\mu,\Sigma}}\rd x
    \int_{\mR^n}\mho_{\mu,\Sigma}\rd x
    =
    D_{\mu,\Sigma}(\mho)+1
$$
which becomes an equality (that is, $D_{\mu,\Sigma}(\mho)=\bR(\mho\|\mho_{\mu,\Sigma})=0$) if and only if $\mho = \mho_{\mu,\Sigma}$. The $\chi^2$-divergence allows the deviation of the actual quantum state of the system from the class of Gaussian states to be described by
\begin{equation}
\label{inf}
    D(\mho)
    :=
    \inf_{\mu\in \mR^n,\ \Sigma\in \mS_n:\ \Sigma\succ 0\ {\rm and}\ \Sigma+i\Theta\succcurlyeq 0}
    D_{\mu,\Sigma}(\mho).
\end{equation}
If the infimum in (\ref{inf}) is achieved, then the appropriate values of $\mu$ and $\Sigma$ specify an ``optimal'' approximation (among Gaussian states) for the actual state of the system. The evolution of this optimal Gaussian state corresponds to an ``effective'' open quantum harmonic oscillator. The  optimal values $\mu_*$ and $\Sigma_*$ can be found as a unique critical point of the $\chi^2$-divergence in (\ref{chi2}) by equating to zero the derivatives
\begin{align}
\label{dchidmu}
    \d_{\mu}D_{\mu,\Sigma}(\mho) = \int_{\mR^n}\mho^2 \d_{\mu}\big(\mho_{\mu,\Sigma}^{-1}\big) \rd x
    & =
    \Sigma^{-1}
    \int_{\mR^n}
    \frac{\mho^2}{\mho_{\mu,\Sigma}} (\mu-x)\rd x,\\
\label{dchidSigma}
    \d_{\Sigma}D_{\mu,\Sigma}(\mho) = \int_{\mR^n}\mho^2 \d_{\Sigma}\big(\mho_{\mu,\Sigma}^{-1}\big) \rd x
    & =
    \frac{1}{2}
    \Sigma^{-1}
    \int_{\mR^n}
    \frac{\mho^2}{\mho_{\mu,\Sigma}} \big(\Sigma-(x-\mu)(x-\mu)^{\rT}\big)\rd x\,
    \Sigma^{-1},
\end{align}
provided $\Sigma_*+i\Theta \succcurlyeq  0$. 
Indeed, in view of the strict concavity  of $\ln\det(\cdot)$ on the set of positive definite matrices \cite{HJ_2007},  for any given $x\in \mR^n$,  the quantity
$$
    \frac{1}{\mho_{\mu,\Sigma}(x)} =
    (2\pi)^{n/2} \exp
    \Big(
        \frac{1}{2}
            |
                \wt{\mu}-\Sigma^{-1/2}x
            |^2 -
            \ln\det (\Sigma^{-1/2})
    \Big)
$$
is a strictly convex function of $(\wt{\mu}, \Sigma^{-1/2})$, where  $\wt{\mu}:= \Sigma^{-1/2}\mu$, and hence, so also is $D_{\mu,\Sigma}(\mho)$ in (\ref{chi2}). Therefore, since there is a smooth bijection between the pairs $(\wt{\mu}, \Sigma^{-1/2})$ and $(\mu,\Sigma)$, the minimization problem  (\ref{inf}) has at most one solution on an open set $\{(\mu,\Sigma)\in \mR^n\x \mS_n:\ \Sigma\succ 0\}$. This solution, when it exists, is necessarily a critical point of $D_{\mu,\Sigma}(\mho)$, and there are no other critical points. If the critical point of $D_{\mu,\Sigma}(\mho)$ satisfies the uncertainty principle condition $\Sigma_*+i\Theta \succcurlyeq  0$, this point also delivers a solution to the constrained problem (\ref{inf}).
Now, the relations (\ref{dchidmu}) and (\ref{dchidSigma}) lead to a fixed-point problem with respect to $\mu_*$ and $\Sigma_*$ described by the coupled nonlinear  vector-matrix equations
\begin{equation}
\label{muSigma*}
    \mu_* = \int_{\mR^n} p_{\mu_*,\Sigma_*}(x)x \rd x,
    \qquad
    \Sigma_* = \int_{\mR^n} p_{\mu_*,\Sigma_*}(x) xx^{\rT}\rd x - \mu_*\mu_*^{\rT},
\end{equation}
whose right-hand sides are the mean vector and the covariance matrix for an auxiliary PDF $p_{\mu_*,\Sigma_*}: \mR^n \to \mR_+$ associated with the QPDF $\mho$ as
\begin{equation}
\label{pmuSigma}
    p_{\mu,\Sigma} = \frac{\mho^2}{(1 + D_{\mu,\Sigma}(\mho))\mho_{\mu,\Sigma}}.
\end{equation}
A different approach to the linearization of nonlinear quantum dynamics has recently been proposed in \cite{VP_2012a} as a quantum Gaussian stochastic linearization method which employs quadratic approximation of Hamiltonians. The study of ``non-Gaussianity'' of quantum states  and their Gaussian approximations based on (\ref{inf}) can benefit from the following dissipation relation for
the $\chi^2$-divergence for fixed parameters $\mu$ and $\Sigma$.

\begin{thm}
\label{th:diss}
Suppose the system-field coupling operators $h_1, \ldots, h_m$ are linear functions of the system variables described by (\ref{N}) and (\ref{Hlin}), and the system Hamiltonian $h_0$ is decomposed according to (\ref{hh}). Also, let the QPDF $\mho$ be continuously differentiable with respect to time and twice continuously differentiable with respect to its spatial variables and satisfy the conditions
\begin{equation}
\label{mhodecay}
  \mho(t,x) = o\Big(\sqrt{\mho_{\mu,\Sigma}(x)} |x|^{-n/2}\Big),
  \qquad
  \d_x \mho(t,x) = o\Big(\sqrt{\mho_{\mu,\Sigma}(x)} |x|^{1-n/2}\Big),
  \qquad
  x \to \infty,
\end{equation}
uniformly over any bounded time interval.
Then the $\chi^2$-divergence  $D_{\mu,\Sigma}(\mho)$ of the actual QPDF $\mho$ from the Gaussian PDF (\ref{gaussWig}) in (\ref{chi2}) satisfies the dissipation relation
\begin{align}
\nonumber
    \d_t D_{\mu,\Sigma}(\mho)
    &
    +(A\mu+2\Theta b)^{\rT}\d_{\mu}D_{\mu,\Sigma}(\mho)
    +
        \Bra
            A\Sigma + \Sigma A^{\rT} + BB^{\rT},\,
            \d_{\Sigma} D_{\mu,\Sigma}(\mho)
        \Ket_{\rF}\\
\label{chidot}
        &
        +
        \underbrace{
        \int_{\mR^n} \frac{|B^{\rT}\d_x\mho|^2}{\mho_{\mu,\Sigma}}\rd x
        -
        \Bra
            BB^{\rT},\, \Sigma^{-1}
        \Ket_{\rF}
        \big(D_{\mu,\Sigma}(\mho) + 1\big)
        }_{\rm nonnegative}
        =
        2
        \Bra
            \frac{\mho}{\mho_{\mu,\Sigma}},\,
            \fG(\mho)
        \Ket.
\end{align}
Here, $A$ and $B$ are the matrices given by (\ref{AB}), and $\fG$ is the integral operator defined by (\ref{fG}).
\hfill $\square$
\end{thm}
\begin{proof}
Since $\mu$ and $\Sigma$ are fixed, $D_{\mu,\Sigma}(\mho)$ depends on time only through the QPDF $\mho$. By differentiating (\ref{chi2}) with respect to time and using the IDE (\ref{mhoIDE}), it follows that
\begin{align*}
\nonumber
    \d_t D_{\mu,\Sigma}(\mho)
    = &
    2
    \int_{\mR^n}
        \frac{\mho\d_t\mho }{\mho_{\mu,\Sigma}}
        \rd x\\
\nonumber
    =&
    2
    \int_{\mR^n}
        \frac{\mho(t,x)}{\mho_{\mu,\Sigma}(x)}
        \Big(\!\!\!
            -\div(\mho(t,x)(Ax+2\Theta b)) + \frac{1}{2}\div^2(\mho(t,x) BB^{\rT})+ \fG(\mho(t,\cdot))(x)
        \!\Big)\rd x\\
\nonumber
    = &
    -(A\mu+2\Theta b)^{\rT}\d_{\mu}D_{\mu,\Sigma}(\mho)
    -
        \Bra
            A\Sigma + \Sigma A^{\rT} + BB^{\rT},\,
            \d_{\Sigma} D_{\mu,\Sigma}(\mho)
        \Ket_{\rF}\\
        &
        +
        \Bra
            BB^{\rT},\, \Sigma^{-1}
        \Ket_{\rF}
        \big(D_{\mu,\Sigma}(\mho) + 1\big)
        - \int_{\mR^n} \frac{|B^{\rT}\d_x\mho(t,x)|^2}{\mho_{\mu,\Sigma}(x)}\rd x
        +
        2
        \Bra
            \frac{\mho}{\mho_{\mu,\Sigma}},
            \fG(\mho)
        \Ket,
\end{align*}
which leads to (\ref{chidot}).
Here, use is made of the divergence theorem in combination with the decay rate conditions (\ref{mhodecay}) and the relations
\begin{align*}
\frac{2\mho}{\mho_{\mu,\Sigma}}
    \div(\mho(Ax+2\Theta b))
    & =
    \div\Big(\frac{\mho^2}{\mho_{\mu,\Sigma}}(Ax+2\Theta b)\Big)
    +
    \frac{\mho^2}{\mho_{\mu,\Sigma}}
    \Tr A
    -
    \mho^2    (Ax+2\Theta b)^{\rT}
    \d_x\big(\mho_{\mu,\Sigma}^{-1}\big)\\
    & =
    \div\Big(\frac{\mho^2}{\mho_{\mu,\Sigma}}(Ax+2\Theta b)\Big)
     +
    (A\mu+2\Theta b)^{\rT}\d_{\mu}\Big(\frac{\mho^2}{\mho_{\mu,\Sigma}}\Big)
    +
        \Bra
            A\Sigma + \Sigma A^{\rT},\,
            \d_{\Sigma} \Big(\frac{\mho^2}{\mho_{\mu,\Sigma}}\Big)
        \Ket_{\rF},
    \\
\frac{\mho}{\mho_{\mu,\Sigma}}
    \div^2(\mho BB^{\rT})
    & =
    \div
    \Big(
        \frac{\mho}{\mho_{\mu,\Sigma}} BB^{\rT} \d_x \mho
    \Big)
    -
    \frac{1}{2}
    \div \big(\mho^2 BB^{\rT} \d_x \big(\mho_{\mu,\Sigma}^{-1}\big)\big)
     +
    \frac{1}{2}
    \mho^2
    \Bra
        BB^{\rT},\,
        \d_x^2 \big(\mho_{\mu,\Sigma}^{-1}\big)
    \Ket_{\rF}
    -
    \frac{|B^{\rT}\d_x \mho|^2}{\mho_{\mu,\Sigma}}\\
    & =
    \frac{1}{2}
    \div
    \big(
        \mho_{\mu,\Sigma}^{-2}
        BB^{\rT}
        \d_x
        \big(
            \mho_{\mu,\Sigma}\mho^2
        \big)
    \big)
    -
    \Bra
        BB^{\rT},\,
        \d_{\Sigma}
        \Big(
            \frac{\mho^2}{\mho_{\mu,\Sigma}}
        \Big)
    \Ket_{\rF}
    +
    \Bra
        BB^{\rT},
        \Sigma^{-1}
    \Ket_{\rF}
    \frac{\mho^2}{\mho_{\mu,\Sigma}}
    -
    \frac{|B^{\rT}\d_x \mho|^2}{\mho_{\mu,\Sigma}}.
\end{align*}
In turn, these relations are obtained from the following identities for the Gaussian PDF $\mho_{\mu,\Sigma}$ in (\ref{gaussWig}):
\begin{align}
\nonumber
(Ax+2\Theta b)^{\rT}
    \d_x\big(\mho_{\mu,\Sigma}^{-1}\big)
    & =
    \mho_{\mu,\Sigma}^{-1}
    \Big(
        (A\mu + 2\Theta b)^{\rT}\xi
        +
        \frac{1}{2}
        \xi^{\rT}
        (A\Sigma + \Sigma A^{\rT})
        \xi
    \Big)\\
\nonumber
    & =
        -(A\mu + 2\Theta b)^{\rT}   \d_{\mu}\big(\mho_{\mu,\Sigma}^{-1}\big)
        +
        \frac{1}{2}
        \Bra
            A\Sigma + \Sigma A^{\rT},\,
            \mho_{\mu,\Sigma}^{-1}\Sigma^{-1}
            -
            2\d_{\Sigma} \big(\mho_{\mu,\Sigma}^{-1}\big)
        \Ket_{\rF}
        \\
\nonumber
    & =
        \frac{\Tr A}{\mho_{\mu,\Sigma}}
        -(A\mu + 2\Theta b)^{\rT}   \d_{\mu}\big(\mho_{\mu,\Sigma}^{-1}\big)
        -
        \Bra
            A\Sigma + \Sigma A^{\rT},\,
            \d_{\Sigma} \big(\mho_{\mu,\Sigma}^{-1}\big)
        \Ket_{\rF},
        \\
\label{dmu}
\d_x
    \big(
        \mho_{\mu,\Sigma}^{-1}
    \big)
    & =
    \mho_{\mu,\Sigma}^{-1}\xi =
    -\d_{\mu}
    \big(
        \mho_{\mu,\Sigma}^{-1}
    \big),\\
\nonumber
\d_x^2
    \big(
        \mho_{\mu,\Sigma}^{-1}
    \big)
    & =
    \mho_{\mu,\Sigma}^{-1}(\Sigma^{-1} + \xi\xi^{\rT})
    =
    2
    \big(
        \mho_{\mu,\Sigma}^{-1} \Sigma^{-1} - \d_{\Sigma}\big(\mho_{\mu,\Sigma}^{-1}\big)
    \big),
\label{dSigma}
    \\
    \d_{\Sigma}
    \big(
        \mho_{\mu,\Sigma}^{-1}
    \big)
    & =
    \frac{1}{2}
    \mho_{\mu,\Sigma}^{-1}
    \big(
        \Sigma^{-1}-\xi\xi^{\rT}
    \big),
\end{align}
 which employ an auxiliary variable $\xi:= \Sigma^{-1}(x-\mu)$,
 where (\ref{dmu}) and (\ref{dSigma}) have already been used in (\ref{dchidmu}) and (\ref{dchidSigma}).
\end{proof}

Although it is not discussed in the proof of Theorem~\ref{th:diss},  the nonnegativeness of the indicated term in (\ref{chidot}) is a corollary of the representation
$$
        \int_{\mR^n} \frac{|B^{\rT}\d_x\mho|^2}{\mho_{\mu,\Sigma}(x)}\rd x
          -
        \Bra
            BB^{\rT},\, \Sigma^{-1}
        \Ket_{\rF}
        \big(D_{\mu,\Sigma}(\mho) + 1\big)
          =
        \Bra
            BB^{\rT},
            \int_{\mR^n}
            \frac{\d_x\mho\d_x\mho^{\rT}}{\mho_{\mu,\Sigma}}\rd x
            -
            \int_{\mR^n}
            \frac{\mho^2}{\mho_{\mu,\Sigma}}\rd x\,
            \Sigma^{-1}
        \Ket_{\rF}
$$
in view of the following lemma which can be regarded as a weighted matrix-valued version of the Dirichlet variational principle \cite{E_1998}.

\begin{lem}
\label{lem:Dir}
Suppose $\varphi: \mR^n \to \mR$ is a twice continuously differentiable function, which is square integrable together with its gradient $\varphi'$ with the weight $\mho_{\mu,\Sigma}^{-1}$ (the reciprocal of the Gaussian PDF from (\ref{gaussWig}) with mean vector $\mu \in \mR^n$ and covariance matrix $\Sigma \succ 0$)  and satisfies
\begin{equation}
\label{phiphi}
    \varphi(x) \varphi'(x) = o\big(\mho_{\mu,\Sigma}(x)|x|^{1-n}\big),
    \qquad
    x\to \infty.
\end{equation}
Then
\begin{equation}
\label{Dir}
        \int_{\mR^n}
        \frac{\varphi'\varphi'^{\rT}}{\mho_{\mu,\Sigma}}\rd x
        \succcurlyeq
        \int_{\mR^n} \frac{\varphi^2}{\mho_{\mu,\Sigma}}
        \rd x \,
        \Sigma^{-1}.
\end{equation}
Moreover, this inequality becomes an equality if and only if $\varphi$ coincides with $\mho_{\mu,\Sigma}$ up to a constant factor. \hfill$\square$
\end{lem}
\begin{proof}
The affine transformation $x\mapsto \mu + \sqrt{\Sigma} x$  of the integration variable  reduces (\ref{Dir}), without loss of generality,  to the case of the standard normal PDF in $\mR^n$ with  $\mu = 0$ and $\Sigma = I_n$.
 In this case,  (\ref{Dir}) is equivalent to the fulfillment of the ``scalar''
 inequality
\begin{equation}
\label{Dir1}
    \Bra
        T,\,
        \int_{\mR^n}
        \frac{\varphi' \varphi'^{\rT}}{\mho_{0,I_n}}\rd x
    \Ket_{\rF}
        \>
        \Tr T
        \int_{\mR^n} \frac{\varphi^2}{\mho_{0,I_n}}
        \rd x
\end{equation}
for any positive definite matrix $T:= (T_{jk})_{1\< j,k\< n}\in \mS_n$. Now, by introducing an auxiliary function
\begin{equation}
\label{psi}
    \psi:= \frac{\varphi}{\sqrt{\mho_{0,I_n}}}
\end{equation}
(which is square integrable together with its gradient $\psi'$) and using the identity $(\sqrt{\mho_{0,I_n}})' = -\frac{1}{2}x \sqrt{\mho_{0,I_n}}$, it follows that
$$
    \varphi' = (\sqrt{\mho_{0,I_n}}\psi)' = \sqrt{\mho_{0,I_n}}\Big(\psi' - \frac{1}{2}x \psi \Big).
$$
Substitution of this equation into (\ref{Dir1}) and integration by parts allows the  left-hand side of the inequality to be represented
as
\begin{equation}
\label{Dir2}
    \Bra
        T,\,
        \int_{\mR^n}
        \frac{\varphi' \varphi'^{\rT}}{\mho_{0,I_n}}\rd x
    \Ket_{\rF}
     =
    \int_{\mR^n}
    \Big\|
        \psi' - \frac{1}{2}x\psi
    \Big\|_T^2 \rd x
    =
    \Bra
        \psi,\,
        \fH(\psi)
    \Ket
    +
    \frac{1}{2}\Tr T \|\psi\|_2^2,
\end{equation}
where $\fH$ is the Hamiltonian of an auxiliary  quantum harmonic oscillator in the position space $\mR^n$ with the stiffness matrix $\frac{1}{2}T$ and mass matrix $\frac{1}{2}T^{-1}$:
\begin{equation}
\label{fH}
    \fH(\psi)
    :=
    \frac{1}{4}
    \|x\|_T^2
    \psi
    -\bra T, \psi''\ket_{\rF},
\end{equation}
with $\psi$ playing the role of a wave function. In (\ref{Dir2}), the divergence theorem has been combined with
the identity
$$
    \Big\|
        \psi' - \frac{1}{2}x\psi
    \Big\|_T^2
     =
    \frac{1}{2}
    \div \big(T((\psi^2)' - \psi^2 x)\big)
    +
    \psi \fH(\psi)
    +
    \frac{1}{2}
    \Tr T
    \psi^2
     =
    \div
    \left(
        \frac{\varphi T\varphi'}{\mho_{0, I_n}}
    \right)
    +
    \psi \fH(\psi)
    +
    \frac{1}{2}
    \Tr T
    \psi^2,
$$
and
use is made of (\ref{psi}) and the decay rate condition (\ref{phiphi}) in the case $\mu=0$ and $\Sigma = I_n$ being considered.
The ground state of the auxiliary oscillator is given by $\psi = \sqrt{\mho_{0,I_n}}$ and does not depend on the matrix $T$, with the ground energy being $\frac{1}{2}\Tr T$ in view of the eigenvalue property $\fH(\sqrt{\mho_{0,I_n}}) = \frac{1}{2}\Tr T\sqrt{\mho_{0,I_n}}$ for the corresponding stationary Schr\"{o}dinger equation. Hence, $\bra \psi, \fH(\psi)\ket \> \frac{1}{2}\Tr T \|\psi\|_2^2 $ for any function $\psi$, which, in combination with (\ref{Dir2}) and (\ref{psi}), leads to
$$
    \Bra
        T,\,
        \int_{\mR^n}
        \frac{\varphi' \varphi'^{\rT}}{\mho_{0,I_n}}\rd x
    \Ket_{\rF}
    \>
    \Tr T \|\psi\|_2^2
    =
    \Tr T
    \int_{\mR^n}
    \frac{\varphi^2}{\mho_{0,I_n}}
    \rd x,
$$
thus establishing (\ref{Dir1}) and (\ref{Dir}) due to arbitrariness of the matrix $T\succ 0$. The second assertion of the lemma follows from the fact that $\psi = \sqrt{\mho_{0,I_n}}$, as a ground state wave function of the Hamiltonian $\fH$ in (\ref{fH}), is unique  up to a constant factor, with the corresponding function $\varphi=\sqrt{\mho_{0,I_n}}\psi = \mho_{0,I_n}$ being the Gaussian PDF in view of (\ref{psi}). \end{proof}

We will now apply Theorem \ref{th:diss} and Lemma~\ref{lem:Dir} to the setting where $\mu$ and $\Sigma$ are evolved so as to remain the unique solution of the optimization problem (\ref{inf}) at every moment of time. It is assumed that $\mu_*\in \mR^n$ and $\Sigma_* \succcurlyeq  -i\Theta$, which deliver   the minimum, are continuously differentiable functions of time described by (\ref{muSigma*}) and (\ref{pmuSigma}). In this case, both $\d_{\mu}D_{\mu,\Sigma}(\mho)$ and $\d_{\Sigma}D_{\mu,\Sigma}(\mho)$ vanish at $\mu=\mu_*$ and $\Sigma= \Sigma_*$, and  the total  time derivative of the corresponding minimum $\chi^2$-divergence $D(\mho) = D_{\mu_*,\Sigma_*}(\mho)$  coincides with the partial time derivative in (\ref{chidot}) which reduces to
\begin{align}
\nonumber
    D(\mho)^{^\centerdot}
    & =
    \left(\d_t D_{\mu,\Sigma}(\mho)
    +
    \dot{\mu}_*^{\rT}
    \d_{\mu}D_{\mu,\Sigma}(\mho)
    +
    \Bra
        \dot{\Sigma}_*,
        \d_{\Sigma}D_{\mu,\Sigma}(\mho)
    \Ket_{\rF}
    \right)
    \Big|_{\mu= \mu_*, \Sigma = \Sigma_*}\\
\label{minchidot}
        &
        =
                \Bra
            BB^{\rT},\, \Sigma_*^{-1}
        \Ket_{\rF}
        (D(\mho) + 1)
        -
        \int_{\mR^n} \frac{|B^{\rT}\d_x\mho|^2}{\mho_{\mu_*,\Sigma_*}}\rd x
        +
        2
        \Bra
            \frac{\mho}{\mho_{\mu_*,\Sigma_*}},\,
            \fG(\mho)
        \Ket
     \<
            2
        \Bra
            \frac{\mho}{\mho_{\mu_*,\Sigma_*}},\,
            \fG(\mho)
        \Ket.
\end{align}
Some remarks are in order in regard to the inner product on the right-hand sides of the dissipation relations (\ref{chidot}) and (\ref{minchidot}). From (\ref{fG}), it follows that
\begin{align}
\nonumber
        \Bra
            \frac{\mho}{\mho_{\mu,\Sigma}},\,
            \fG(\mho)
        \Ket
        & =
        -2
        \int_{\mR^n}
        \frac{\mho(x)}{\mho_{\mu,\Sigma}(x)}
    \int_{\mR^d}
    \Pi(x,v)\mho(x-\Theta Z^{\rT}v)\rd v
    \rd x\\
\label{ttt}
    & = -2
        \int_{\mR^n\x \mR^d}
        \wt{\Pi}(x,v)
        \wt{\mho}(x) \wt{\mho}(x-\Theta Z^{\rT}v)
    \rd x
    \rd v ,
\end{align}
where the time argument of the QPDF $\mho$ is omitted for brevity. Here, 
$\wt{\mho}:= \frac{\mho}{\sqrt{\mho_{\mu,\Sigma}}}$ is an auxiliary function  whose $L^2$-norm is related to the $\chi^2$-divergence in (\ref{chi2}) as
\begin{equation}
\label{DDD}
    \|\wt{\mho}\|_2 = \sqrt{D_{\mu,\Sigma}(\mho) + 1},
\end{equation}
and
\begin{equation}
\label{Pit}
    \wt{\Pi}(x,v)
    :=
    \Pi(x,v)\sqrt{\frac{\mho_{\mu,\Sigma}(x-\Theta Z^{\rT}v)}{\mho_{\mu,\Sigma}(x)}}
    =
    \Pi(x,v)
    \re^{\frac{1}{2}(x-\mu)^{\rT} \Sigma^{-1} \Theta Z^{\rT} v  - \frac{1}{4} \|\Theta Z^{\rT} v\|_{\Sigma^{-1}}^2}.
\end{equation}
However, the absence of decay in the kernel function $\Pi(x,v)$ as $x \to \infty$ in (\ref{Pi}) makes the following upper bound (which employs only the Cauchy-Bunyakovsky-Schwarz inequality and (\ref{DDD})) for the  right-hand side of (\ref{ttt}) ineffective:
\begin{align*}
    \Big|
        \int_{\mR^n\x \mR^d}
        \wt{\Pi}(x,v)
        \wt{\mho}(x) \wt{\mho}(x-\Theta Z^{\rT}v)
    \rd x
    \rd v
    \Big|
    & \<
        \int_{\mR^d}
        \|
        \wt{\Pi}(\cdot,v)
        \|_{\infty}
        \int_{\mR^n}
        \big|
        \wt{\mho}(x)
        \wt{\mho}(x-\Theta Z^{\rT}v)
        \big|
        \rd x
        \rd v\\
        & \<
        \int_{\mR^d}
        \|\wt{\Pi}(\cdot,v)\|_{\infty}
        \rd v\,
        \big(
            D_{\mu,\Sigma}(\mho) + 1
        \big),
\end{align*}
because, in view of (\ref{Pit}), $\|\wt{\Pi}(\cdot,v)\|_{\infty} := \sup_{x\in \mR^n} |\wt{\Pi}(x,v)|=+\infty$ for any $v\in \mR^d\setminus \{0\}$  such that $\wt{H}_0(v)\ne 0$. Therefore, nontrivial estimates for (\ref{ttt}) should be based on a more subtle analysis using the information on smoothness of the QPDF $\mho$ as mentioned in Section \ref{sec:Hlin}.

\section{Conclusion}\label{sec:conc}

We have considered a class of open quantum stochastic systems, whose dynamic variables satisfy CCRs and are governed by Markovian Hudson-Parthasa\-rathy QSDEs,  with the Hamiltonian and coupling operators represented in the Weyl quantization form.  In extending the Wigner-Moyal approach from isolated systems to open quantum stochastic systems, we have obtained an IDE for the evolution of the QCF which encodes the moment dynamics of the system variables. A related IDE, which governs the QPDF dynamics,  coincides with the classical FPKE in the case of open quantum harmonic oscillators and becomes the Moyal equation for isolated quantum systems. For a class of open quantum systems with linear system-field coupling and a nonquadratic Hamitonian, the IDE for the QPDF consists of an FPKE part and a Moyal term, which leads to non-Gaussian dynamics and negative values of the QPDF. The smoothness of fundamental solutions of this IDE needs a separate research into an appropriate counterpart of H\"{o}rmander conditions.
We have discussed an approximate computation of invariant QPDFs in the presence of Gaussian-shaped potentials
and the deviation of the system from Gaussian quantum states in terms of the $\chi^2$-divergence 
of the QPDF.
The results of the paper may find applications to different aspects of relaxation dynamics in open quantum stochastic systems, such as  the existence and phase-space representation of invariant states and the rates of convergence to them.

\end{document}